%% file: main_TIT.tex
\pdfoutput=1
\documentclass[journal]{IEEEtran}

\usepackage{amsmath,amssymb,amsfonts,amsthm}
\usepackage{algorithm}
\usepackage{algpseudocode}
\usepackage{graphicx}
\usepackage{booktabs}
\usepackage{cite}
\usepackage{url}
\usepackage{bm}
\usepackage{mathtools}
\usepackage{xcolor}
\usepackage{dblfloatfix}
\newcommand{\tabIVwidth}{\columnwidth}
\usepackage{tikz}
\usetikzlibrary{arrows.meta,positioning,calc}

\newtheorem{theorem}{Theorem}
\newtheorem{lemma}{Lemma}
\newtheorem{proposition}{Proposition}

\theoremstyle{definition}
\newtheorem{definition}{Definition}

\newcommand{\R}{\mathbb{R}}
\newcommand{\E}{\mathbb{E}}
\newcommand{\I}{\mathrm{I}}
\newcommand{\HH}{\mathrm{H}}
\newcommand{\Cdel}{C(d)}
\newcommand{\db}{\bar d}
\newcommand{\ones}{\mathbf{1}}
\newcommand{\Bin}{\beta}
\newcommand{\RDM}{\mathcal{R}_{\mathrm{DM}}}
\newcommand{\RVTR}{\mathcal{R}_{\mathrm{VTR}}}
\newcommand{\hg}{\eta}

\input{data_TIT/numbers.tex}

\begin{document}

\title{Improved Lower Bounds on the Capacity of the Binary Deletion Channel
via a Learning Approach to Run-Length Inputs}

\author{Hassan Khodaiemehr, Chen Feng, and Tolga M.~Duman%
\thanks{H.~Khodaiemehr and C.~Feng are with the School of Engineering, The University of British Columbia,
Okanagan Campus, Kelowna, BC, Canada (e-mail: hassan.khodaiemehr@ubc.ca; chen.feng@ubc.ca).}%
\thanks{T.~M.~Duman is with the Department of Electrical and Electronics Engineering,
Bilkent University, Ankara 06800, Turkey (e-mail: duman@ee.bilkent.edu.tr).}%
\thanks{The work of C.~Feng was supported in part by the NSERC Discovery
Grants RGPIN-2023-04962.
The work of T.~M.~Duman was funded by the European Union through the
ERC Advanced Grant 101054904: TRANCIDS.
Views and opinions expressed are, however, those of the authors only and
do not necessarily reflect those of the European Union or the European
Research Council Executive Agency.
Neither the European Union nor the granting authority can be held
responsible for them.}%
\thanks{This paper was presented in part at the Canadian Workshop on Information Theory (CWIT), 2024. An earlier version of this paper is available on arXiv~\cite{khodaiemehr2026improvedlowerboundscapacity}.}
}

\maketitle

\begin{abstract}
The best constructive lower bounds on the capacity of the binary deletion channel come from
random codes with independent run lengths, yet the two strongest such bounds, due to Drinea and Mitzenmacher and to Venkataramanan, Tatikonda, and Ramchandran, have
been evaluated mainly for geometric or low-parameter run-length laws.
We let a learning algorithm choose the run-length law freely, which raises three challenges.
First, the Drinea--Mitzenmacher functional is an infinite sum over the ways deletions merge runs; we
show that it depends on the law only through its mean, its full-deletion probability, and bilinear
forms in the law and its renewal weights, so that gradients of truncations are exact and every truncation can only
lower the bound.
Second, for non-Markov inputs the output is no longer Markov and the
Venkataramanan--Tatikonda--Ramchandran analysis breaks down; we show that the residual length of the
current input run turns the output into a hidden Markov chain, which extends the bound to every
finite-support law.
Third, its correction term counts only output runs formed from three input runs; we prove a larger
correction that accounts for every output run formed by several input runs, which improves the
published bound even for truncated geometric laws.
Certified by interval arithmetic, the new bounds exceed all previously published
deterministic lower bounds at every tabulated deletion probability, by up to $\AbsGainVal$~bits per
channel use and $\RelGainVal\%$.
At large deletion probability the learned laws concentrate on run-length clusters with survivor
counts about two standard deviations apart, like a pulse-amplitude constellation.
\end{abstract}


\section{Introduction}
\label{sec:intro}
\IEEEPARstart{S}{ynchronization} channels---deletions, insertions, duplications, and more general timing
errors---remain among the most challenging models in information
theory~\cite{Mitzenmacher2009,CheraghchiCrawford2020,Mercier2010,MercierTarokhLabeau2012,IyengarSiegelWolf2016}.
The independent and identically distributed (i.i.d.) binary deletion channel with deletion probability $d\in[0,1]$ (BDC$_d$) deletes each
transmitted bit independently with probability $d$ and concatenates the surviving bits without
revealing deletion locations. Unlike the binary erasure channel, the receiver observes a
\emph{variable-length} subsequence of the input and obtains no side information indicating where
deletions occurred. This combination of random output length and unknown alignment places BDC$_d$
outside the classical discrete memoryless channel (DMC) toolkit, and standard single-letter
characterizations do not apply.
Deletions arise in magnetic and optical recording, packet networks with lost timing, DNA/RNA
storage and sequencing, and high-speed links subject to bit slips.

A coding theorem and weak converse for a broad class of discrete channels with synchronization
errors were established by Dobrushin~\cite{Dobrushin1967}, building on his information-spectrum
formulation of Shannon's theorem~\cite{Dobrushin1959,Shannon1948}.
Information stability implies that the capacity $\Cdel$ exists and equals the normalized limit of
finite-block mutual information,
\begin{equation}
\label{eq:C-limit}
\Cdel
\;=\;
\lim_{N\to\infty}\frac1N\max_{p_{X^N}}\I(X^N;Y),
\end{equation}
yet the theorem does not yield a closed-form expression or a tractable numerical procedure for
$\Cdel$~\cite{Mitzenmacher2009,CheraghchiCrawford2020}.
Information stability has more recently been established for synchronization errors governed by
stationary ergodic finite-state Markov chains, covering insertions, deletions, and
substitutions (IDS)~\cite{MorozovDuman2024MarkovISIT,MorozovDuman2026Markov}.
Gallager~\cite{Gallager1961} gave early lower bounds on channel capacity by sequential decoding
of convolutional codes over IDS channels.
Diggavi and Grossglauser related deletion models to finite-buffer losses and showed that deletion
capacity remains within a controlled gap of erasure capacity under mismatched
decoding~\cite{DiggaviAllerton2001,Diggavi2007}.
Random coding with alternating runs whose lengths are drawn i.i.d.\ from a block-length law
(often geometric) subsequently produced strong constructive lower bounds, including the
celebrated estimate $\Cdel>0.1185(1-d)$~\cite{MitzenmacherDrinea2006,DrineaMitzenmacher2006LB,DrineaMitzenmacher2007,KirschDrinea2009},
later improved by large-scale computation~\cite{RubinsteinCon2023}.
Those constructions, and the present paper, operate at infinite blocklength: a stationary renewal
input is fed to BDC$_d$ and an achievable rate is read from an explicit functional of the
run-length law~$P$.

Nontrivial upper bounds arrived comparatively late.
Diggavi, Mitzenmacher, and Pfister~\cite{DiggaviMitzenmacherPfister2007} obtained the first
nontrivial upper bounds via genie-aided run information.
Fertonani and Duman~\cite{FertonaniDuman2010} substantially tightened numerical upper (and some lower)
bounds by evaluating auxiliary memoryless channels with the Blahut--Arimoto (BA)
algorithm~\cite{Blahut1972,Arimoto1972}.
Rahmati and Duman~\cite{RahmatiDuman2015}
improved fragmentation upper bounds for $d\in(0.65,1)$, obtaining in particular
$\Cdel\le 0.4143(1-d)$ on that interval.
A related fragmentation argument yields a nontrivial upper bound on nonbinary deletion capacity
via the inequality $C_{2K}(d)\le C_2(d)+(1-d)\log K$~\cite{RahmatiDuman2013Nonbinary}.
For small deletion probability, Kalai, Mitzenmacher, and Sudan showed
$\Cdel=1-\Theta(d\log(1/d))$~\cite{KalaiSudan2010}, while Kanoria and Montanari computed a three-term
series expansion of $\Cdel$ together with near-optimal inputs at that order~\cite{KanoriaMontanari2013}.
In the joint small-deletion/small-substitution regime, the i.i.d.\ deletion/substitution capacity
admits the approximation $1-H(p_d)-H(p_s)$~\cite{KazemiDuman2025DelSub}.
Parallel small-noise expansions for binary insertion channels (simple and Gallager insertion
models) establish capacity approximations that differ from Bernoulli$(1/2)$ rates only in higher-order
terms~\cite{TeginDuman2025InsertionISIT,TeginDuman2026InsertionCap}.
Cheraghchi~\cite{Cheraghchi2019} derived the first fully explicit, computer-free nontrivial upper
bounds for a wide range of~$d$, including forms involving the golden ratio
$\phi=(1+\sqrt5)/2$.
Computational BA upper bounds were subsequently tightened by Rubinstein and
Con~\cite{RubinsteinCon2023} (to $n\le 28$), and further by Pinto and
Ribeiro~\cite{PintoRibeiro2026ParallelBA} via a GPU-parallelized Blahut--Arimoto procedure on the
exact finite-length deletion DMCs $\mathrm{BDC}_{n,k}$.
Combined with the Rahmati--Duman fragmentation lemma, their bound at $d'=0.64$ yields the improved
high-noise statement $\Cdel\le 0.3578(1-d)$ for all $d\ge 0.64$.
These $C_{n,k}$ computations yield upper bounds on $\Cdel$~\cite{FertonaniDuman2010}.

A concurrent manuscript of Papailiopoulos~\cite{Papailiopoulos2026}, developed with
interacting language-model agents as disclosed therein, gives a computer-assisted
enclosure $L(d)\le\Cdel\le U(d)$ whose half-width is at most $0.0095$ bits uniformly in~$d$,
with mean radius below $0.00652$.
The converse side uses a stationary-source reduction together with two finite-window tests that
control unobserved input beyond a computational horizon.
The achievability side evaluates specified finite-state sources (output entropy minus deletion and
residual-mask entropy) and independent-run sources (the Drinea--Mitzenmacher jigsaw rate plus a
finite lower estimate of run-start uncertainty, in the spirit of
Kirsch--Drinea~\cite{KirschDrinea2009} with disjoint counting in place of simulation).
Pointwise bounds are transported to every $d\in[0,1]$ by the monotonicity of
$\Cdel/(1-d)$~\cite{RahmatiDuman2015}.
That programme is complementary to the present work: it produces both lower and upper bounds, and a uniform
error certificate, from finite inequalities with directed arithmetic.
It does not optimize the Drinea--Mitzenmacher or Venkataramanan functionals over a free run-length
law~$P$, and it does not identify the sparse comb inputs that appear at large~$d$ below.
We compare the two lower envelopes in Section~\ref{sec:results}. The results of
\cite{Papailiopoulos2026} are stronger on most of the interior of $[0,1]$, while the present
lower bound is larger at low and high deletion probabilities ($d=0.01$, $d=0.80$ and $d=0.90$).

Surveys of Mitzenmacher~\cite{Mitzenmacher2009} and Cheraghchi and
Ribeiro~\cite{CheraghchiCrawford2020} give comprehensive accounts; related synchronization models
(sticky channels, timing errors, combined sync/substitution noise) are treated
in~\cite{MitzenmacherSticky2008,MercierTarokhLabeau2012,IyengarSiegelWolf2016}.
Deep-learning detectors for concatenated marker codes over deletion and insertion
channels address practical decoding rather than capacity-achieving input
optimization~\cite{KargiDuman2024DL,KargiDuman2026SymbolLevel}.
Burst, segmented, and mixed deletion/insertion/substitution
codes~\cite{LiuDuman2025Segmented,WangYaakobiDuman2026Burst,LiuDuman2026TwoBurstsDI,LiuCaiDuman2026DelSubBurst},
together with DNA-motivated marker and composite-alphabet
designs~\cite{HaghighatDuman2025HalfMarker,TeginDuman2026CompositeDNA}, are likewise complementary
to the binary i.i.d.\ deletion capacity studied here.

The run-length functionals of Drinea--Mitzenmacher (DM) and of VTR are achievable rates for inputs
whose runs have i.i.d.\ lengths with law~$P$: the DM functional is the rate of a random code with
jigsaw decoding, and the VTR functional bounds the information rate of a stationary renewal input.
Both were evaluated only on low-dimensional families: geometric laws (Markov inputs) and the
two-parameter family of~\cite{DrineaMitzenmacher2007}.
We treat $P$ as a free distribution on a finite support and make three contributions.
\begin{enumerate}
\item \emph{DM functional.}
We show that the functional depends on $P$ only through its mean, its full-deletion probability, and
bilinear forms in $P$ and the renewal weights $\Psi$ of the law, that the received run-length law is a
mixture of binomial laws with weights proportional to $\Psi$, and that its mean has a closed form
(Lemma~\ref{lem:dm-reduce}).
Truncation of every infinite sum is one-sided (Proposition~\ref{prop:dm-trunc}).
\item \emph{VTR functional.}
We extend the VTR decomposition from Markov inputs to every finite-support run-length law.
The output of a renewal input is not Markov, but the residual length of the current input run turns
it into a hidden Markov chain (Lemma~\ref{lem:residual}), whose conditional entropies are computed
exactly by a finite filter (Proposition~\ref{prop:filter}).
We prove a new lower bound on the correction term of the decomposition, which accounts for every
output run formed from several input runs (Lemma~\ref{lem:chain}) and exceeds the VTR correction
by at least the factor $1+D$ (Proposition~\ref{prop:vtr-compare}).
The resulting bound (Theorem~\ref{thm:vtr}) improves the published VTR bound even when the input is
restricted to geometric run-length laws truncated to a finite support.
\item \emph{Optimization and verified evaluation.}
Both functionals are maximized over the simplex by gradient methods, and every reported number is
the lower endpoint of an interval-arithmetic enclosure computed with directed rounding and a
verified logarithm (Section~\ref{sec:numerics}, Proposition~\ref{prop:verified}).
The resulting lower bounds exceed all previously published deterministic lower bounds at every
deletion probability for which run-length bounds were tabulated (Table~\ref{tab:comparison}).
At large $d$ the optimized laws concentrate on a few separated clusters of run lengths whose
survivor counts are about two binomial standard deviations apart, a structure that neither the
geometric family nor the two-parameter family of~\cite{DrineaMitzenmacher2007} contains.
\end{enumerate}

The optimization procedure selects only the run-length law $P$, while the validity of the resulting bounds follows from Theorems~\ref{thm:dm} and~\ref{thm:vtr}, Proposition~\ref{prop:dm-trunc}, and the verified enclosure procedure of Proposition~\ref{prop:verified}.
Compared with the earlier version of this work~\cite{khodaiemehr2026improvedlowerboundscapacity}, the present paper makes three main advances. First, the correction term previously used for the VTR functional is replaced by the chain penalty of Lemma~\ref{lem:chain}, yielding a fully rigorous formulation. Second, all numerical evaluations are carried out using verified interval arithmetic rather than double-precision arithmetic. Third, complete proofs are provided for all auxiliary results, and the presentation has been reorganized to provide the supporting arguments in a more detailed and systematic manner.

The rest of this paper is organized as follows. Section~\ref{sec:model} sets up the model and proves the properties of stationary renewal inputs
and their renewal weights.
Sections~\ref{sec:dm} and~\ref{sec:vtr} treat the two functionals.
Section~\ref{sec:numerics} describes the optimization and the verified evaluation, and
Section~\ref{sec:results} reports the bounds.
Section~\ref{sec:discussion} discusses open problems, and Section~\ref{sec:conclusion} concludes the paper.

\input{sec_model_TIT}
\input{sec_dm_TIT}
\input{sec_vtr_TIT}
\input{sec_numerics_TIT}
\input{sec_results_TIT}

\input{sec_discussion_TIT}

\section{Conclusion}
\label{sec:conclusion}

We treated the run-length law of the Drinea--Mitzenmacher and Venkataramanan--Tatikonda--Ramchandran
achievability functionals as a free distribution.
For the first functional we obtained an exact bilinear reduction, a binomial-mixture form of the
received run-length law, and one-sided truncation.
For the second we extended every term to general finite-support laws and proved a new lower bound on
the correction term, which already improves the published bound for truncated geometric laws.
Gradient search over the simplex, followed by interval-arithmetic evaluation, gives lower bounds
that exceed all previously published deterministic lower bounds at every deletion probability for
which run-length bounds were tabulated.
The concurrent enclosure of~\cite{Papailiopoulos2026} has larger lower endpoints at
$d\in\{\PPbetter\}$, while the present bounds are larger at $d\in\{\WeBetter\}$.
Several questions remain open, including tighter bounds on the correction term, the structure of
maximizing laws at large $d$, and the extension to other synchronization channels.

\bibliographystyle{IEEEtran}
\bibliography{refs_TIT}

\end{document}

%% file: data_TIT/numbers.tex
\newcommand{\AbsGainVal}{6.86\times10^{-3}}
\newcommand{\AbsGainD}{0.30}
\newcommand{\RelGainVal}{6.8}
\newcommand{\RelGainD}{0.90}
\newcommand{\VtrUpTo}{0.80}
\newcommand{\DmFrom}{0.85}
\newcommand{\NumTab}{19}
\newcommand{\MeanGap}{2.11}
\newcommand{\GapMin}{1.64}
\newcommand{\GapMax}{2.86}
\newcommand{\NGaps}{20}
\newcommand{\CentEight}{10.48, 29.95, 56.31}
\newcommand{\ClusterMassMin}{83}
\newcommand{\ClusterMassMax}{92}
\newcommand{\KappaGain}{2.4}
\newcommand{\GapEight}{2.24\, \textrm{and}\, 2.03}
\newcommand{\MarkovGainMin}{0.1}
\newcommand{\MarkovGainMax}{4.7}
\newcommand{\FreeGainMax}{2.5}
\newcommand{\DmFreeMaxD}{0.60}
\newcommand{\DmFreeMax}{2.41}
\newcommand{\DmFreeRelNine}{6.7}
\newcommand{\PPbetter}{0.05,\allowbreak 0.1,\allowbreak 0.2,\allowbreak 0.3,\allowbreak 0.4,\allowbreak 0.5,\allowbreak 0.6,\allowbreak 0.7}
\newcommand{\WeBetter}{0.01,\allowbreak 0.8,\allowbreak 0.9}
\newcommand{\MCphiA}{0.0060}
\newcommand{\MCestA}{0.0085}
\newcommand{\MCciA}{0.0011}
\newcommand{\MCphiB}{0.0465}
\newcommand{\MCestB}{0.0669}
\newcommand{\MCciB}{0.0038}
\newcommand{\MCphiC}{0.1411}
\newcommand{\MCestC}{0.1929}
\newcommand{\MCciC}{0.0044}
\makeatletter
\@namedef{env@0.01}{0.92215}
\@namedef{env@0.02}{0.86456}
\@namedef{env@0.03}{0.81478}
\@namedef{env@0.04}{0.77017}
\@namedef{env@0.05}{0.72946}
\@namedef{env@0.06}{0.69191}
\@namedef{env@0.07}{0.65703}
\@namedef{env@0.08}{0.62447}
\@namedef{env@0.09}{0.59396}
\@namedef{env@0.1}{0.56530}
\@namedef{env@0.15}{0.44469}
\@namedef{env@0.2}{0.35332}
\@namedef{env@0.25}{0.28335}
\@namedef{env@0.3}{0.22936}
\@namedef{env@0.35}{0.18727}
\@namedef{env@0.4}{0.15403}
\@namedef{env@0.45}{0.12739}
\@namedef{env@0.5}{0.10576}
\@namedef{env@0.55}{0.08787}
\@namedef{env@0.6}{0.07277}
\@namedef{env@0.65}{0.05978}
\@namedef{env@0.7}{0.04837}
\@namedef{env@0.75}{0.03821}
\@namedef{env@0.8}{0.02901}
\@namedef{env@0.85}{0.02070}
\@namedef{env@0.9}{0.01322}
\@namedef{env@0.95}{0.00634}
\newcommand{\env}[1]{\@nameuse{env@#1}}
\makeatother

%% file: sec_model_TIT.tex
\section{Model and Renewal Inputs}
\label{sec:model}

Table~\ref{tab:notation} collects the main symbols and notation. Logarithms are to base $2$ unless written $\ln$.
For a probability vector $q$ on a countable set, $\HH(q)=-\sum_k q_k\log q_k$ with $0\log0=0$, and
$h(p)=\HH(p,1-p)$ is the binary entropy function.
We write $\db=1-d$, use the convention $\binom{a}{k}=0$ unless $0\le k\le a$, and write
\[
\Bin_a(k)=\binom{a}{k}\db^{\,k}d^{\,a-k},\qquad a,k\ge0,
\]
for the probability that exactly $k$ of $a$ bits survive when each bit is deleted independently with
probability $d$.
For a finite or countable random vector $A$, $A^t=(A_1,\dots,A_t)$ and $A_s^t=(A_s,\dots,A_t)$.
The indicator of an event $E$ is $\ones_E$ or $\ones\{E\}$.

\begin{table}[!t]
\caption{Main notation}
\label{tab:notation}
\centering
\footnotesize
\begin{tabular}{@{}ll@{}}
\toprule
Symbol & Meaning \\
\midrule
$d$, $\db=1-d$ & deletion and survival probabilities \\
$P=(P_z)_{z=1}^{Z}$ & run-length law of the input, $P_Z>0$ \\
$\mu$ & mean run length $\sum_z zP_z$ \\
$D$ & probability that a run is deleted, $\sum_zP_zd^{\,z}$ \\
$L_j$, $N_j$ & length of input run $j$; number of its surviving bits \\
$R_n$ & number of input runs that meet the window $\{1,\dots,n\}$ \\
$\ell_j$, $\tilde N_j$ & length of run $j$ inside the window; its survivors there \\
$\Bin_a(k)$ & $\binom{a}{k}\db^{\,k}d^{\,a-k}$ \\
$Q_{r,i}$ & probability that $i$ i.i.d.\ runs have total length $r$ \\
$\Psi_r$ & renewal weight $\sum_{i\ge0}D^iQ_{r,i}$ \\
$q=(q_k)$ & law of the length of an output run \\
$\rho$ & output runs per input bit, see~\eqref{eq:rho} \\
$U_t$ & residual length of the input run after output bit $t$ \\
$S_t$ & runs deleted entirely between output bits $t-1$, $t$ \\
$\pi$, $\nu$ & residual laws, see~\eqref{eq:pi} and~\eqref{eq:nu} \\
$\hg(z,r,s)$ & entropy of the hypergeometric split~\eqref{eq:hsplit} \\
$\Theta_{z,r}$ & chain kernel~\eqref{eq:Kchain} \\
\bottomrule
\end{tabular}
\end{table}

An input $X^n\in\{0,1\}^n$ is sent through BDC$_d$.
Let $\Delta_1,\Delta_2,\dots$ be i.i.d.\ Bernoulli$(d)$ random variables, independent of the input,
where $\Delta_k=1$ means that bit $k$ is deleted.
The output is $Y=(X_k:\,k\le n,\ \Delta_k=0)$, listed in increasing order of $k$, and its length is
$M_n=\sum_{k=1}^n(1-\Delta_k)\sim\mathrm{Bin}(n,\db)$.
For every $n$ and every input law, $\frac1n\I(X^n;Y)\le\frac1n\max_{p_{X^n}}\I(X^n;Y)$, and the
right-hand side converges to $\Cdel$ by~\eqref{eq:C-limit}.
Hence every sequence of input laws satisfies
\begin{equation}
\label{eq:liminf-C}
\liminf_{n\to\infty}\frac1n\I(X^n;Y)\le\Cdel .
\end{equation}
All lower bounds in this paper are obtained by bounding the left-hand side of~\eqref{eq:liminf-C}
from below for a stationary renewal input.

\subsection{Stationary renewal inputs}

A \emph{run} is a maximal constant substring.
Throughout, $P=(P_z)$ is a probability law on $\{1,\dots,Z\}$ with $P_Z>0$, and $L$ denotes a
generic random variable with law $P$.
We write
\begin{equation}
\label{eq:mu-D}
\mu=\sum_{z=1}^{Z}zP_z,
\qquad
D=\sum_{z=1}^{Z}P_zd^{\,z}.
\end{equation}
Thus $\mu$ is the mean run length and $D$ is the probability that all bits of a run of length $L$
are deleted.
Since $0<d<1$, we have $0<P_Zd^{\,Z}\le D\le d<1$.

\begin{definition}[Stationary renewal input]
\label{def:renewal}
Let $L_1$ have the law
\begin{equation}
\label{eq:first-run}
\Pr[L_1=u]=\frac{\Pr[L\ge u]}{\mu},\qquad u=1,\dots,Z,
\end{equation}
let $L_2,L_3,\dots$ be i.i.d.\ with law $P$, and let $B$ be a uniform bit, all independent.
Let $\theta_0=0$ and $\theta_j=L_1+\dots+L_j$.
The \emph{stationary renewal input} with run-length law $P$ is the sequence $X_1,X_2,\dots$ with
$X_k=B\oplus((j-1)\bmod2)$ for $\theta_{j-1}<k\le\theta_j$.
\end{definition}

Run $j$ occupies the positions $\theta_{j-1}+1,\dots,\theta_j$, its length is $L_j$, and the symbols of
consecutive runs alternate.
The law~\eqref{eq:first-run} sums to one because $\sum_{u\ge1}\Pr[L\ge u]=\E[L]=\mu$.

\begin{lemma}[Stationarity]
\label{lem:stationary}
The sequence $(X_k)_{k\ge1}$ of Definition~\ref{def:renewal} is stationary.
For every $k\ge1$, the bit $X_k$ is uniform, the number $V_k$ of positions after $k$ that belong
to the run containing $k$ has the law
\begin{equation}
\label{eq:pi}
\pi(u)=\frac{\Pr[L>u]}{\mu},\qquad u=0,\dots,Z-1,
\end{equation}
and $X_k$ and $V_k$ are independent.
\end{lemma}

\begin{proof}
The shifted sequence $X'_k=X_{k+1}$, $k\ge1$, is obtained from Definition~\ref{def:renewal} with
the variables $(L'_1,B',\allowbreak L'_2,L'_3,\dots)$ defined as follows.
If $L_1\ge2$, then $L'_1=L_1-1$, $B'=B$, and $L'_j=L_j$ for $j\ge2$.
If $L_1=1$, then $L'_1=L_2$, $B'=1-B$, and $L'_j=L_{j+1}$ for $j\ge2$.
On each of the two events $\{L_1\ge2\}$ and $\{L_1=1\}$, which depend on $L_1$ alone, the variables
$L'_1$, $B'$, and $(L'_j)_{j\ge2}$ are conditionally independent, $B'$ is conditionally uniform, and
$(L'_j)_{j\ge2}$ is conditionally i.i.d.\ with law $P$.
The conditional laws of $B'$ and of $(L'_j)_{j\ge2}$ are the same on both events; hence, unconditionally,
$L'_1$, $B'$, and $(L'_j)_{j\ge2}$ are independent, $B'$ is uniform, $(L'_j)_{j\ge2}$ is i.i.d.\ with
law $P$, and it remains to identify the law of $L'_1$.
Since $\Pr[L_1=1]=1/\mu$, for $u\ge1$
\begin{align*}
\Pr[L'_1=u]&=\Pr[L_1=u+1]+\Pr[L_1=1]P_u\\
&=\frac{\Pr[L\ge u+1]+P_u}{\mu}=\frac{\Pr[L\ge u]}{\mu}.
\end{align*}
Thus $(L'_1,B',L'_2,\dots)$ has the same law as $(L_1,B,L_2,\dots)$, and $(X'_k)$ has the same law as
$(X_k)$, which is stationarity.
For $k=1$ we have $X_1=B$ and $V_1=L_1-1$, which are independent, and
$\Pr[L_1-1=u]=\Pr[L\ge u+1]/\mu=\pi(u)$.
Since $(X_k,V_k)$ is the same function of $(X_k,X_{k+1},\dots)$ for every $k$ (the number of
consecutive positions after $k$ that carry the symbol $X_k$), stationarity extends these properties
to every $k$.
\end{proof}

The input of blocklength $n$ is the prefix $X^n$.
It is the window of the stationary process, and its entropy rate is~\cite{Diggavi2007,VenkataramananTatikondaRamchandran2013}
\begin{equation}
\label{eq:entropy-rate}
\lim_{n\to\infty}\frac1n\HH(X^n)=\frac{\HH(P)}{\mu}.
\end{equation}
For the geometric law $P_z=(1-p)p^{z-1}$ on all positive integers, which gives a first-order Markov
input, $\mu=1/(1-p)$ and the right-hand side equals $h(p)$.
All laws used in this paper have finite support; truncated geometric laws are used as a Markov
control in Section~\ref{sec:results}.

\subsection{The window and deletions at the level of runs}
\label{subsec:runs}

Let $R_n=\min\{j:\theta_j\ge n\}$ be the number of runs that meet the window $\{1,\dots,n\}$, and let
$\ell_j=L_j$ for $j<R_n$ and $\ell_{R_n}=n-\theta_{R_n-1}$ be the lengths of these runs inside the window.
Runs $2,\dots,R_n-1$ lie entirely inside the window; runs $1$ and $R_n$ may be cut by its boundary.
Given $n$, the window $X^n$ and the vector $(B,\ell_1,\dots,\ell_{R_n})$ determine each other.

Let $N_j$ be the number of surviving bits of run $j$, and $\tilde N_j$ the number of surviving
bits of run $j$ inside the window, for $j\le R_n$.
Then $\tilde N_j=N_j$ for $j<R_n$.
The numbers of survivors depend on disjoint sets of deletion indicators, which are independent of the
input.
Hence:
\begin{enumerate}
\item conditionally on $X^n$, the variables $\tilde N_1,\dots,\tilde N_{R_n}$ are independent with
$\tilde N_j\sim\mathrm{Bin}(\ell_j,\db)$;
\item the pairs $(L_j,N_j)$, $j\ge2$, are i.i.d.\ with $\Pr[L_j=z,N_j=m]=P_z\Bin_z(m)$, and
independent of $(L_1,N_1,B)$.
\end{enumerate}
The output is the concatenation, over $j=1,\dots,R_n$, of $\tilde N_j$ copies of the symbol of run
$j$.
Deletions shorten runs or delete them entirely; they never split a run, but runs of the same symbol
merge in the output when all runs between them are deleted.

The following notion is used by Drinea and Mitzenmacher~\cite{DrineaMitzenmacher2007}.

\begin{definition}[Group of an output run]
\label{def:group}
Let $o$ be a run of $Y$ and let $j$ be the input run that contains the first bit of $o$.
The \emph{group} of $o$ consists of the input runs $j,j+1,\dots,j+2i$, where $j+2i+1$ is the first
run after $j$ that has the opposite symbol and at least one surviving bit.
The integer $i\ge0$ is the \emph{type} of $o$.
The runs $j+1,j+3,\dots,j+2i-1$ are deleted entirely, and $o$ consists of the surviving bits of
the same-symbol runs $j,j+2,\dots,j+2i$.
\end{definition}

In a group, the first run has at least one surviving bit, and the other same-symbol runs may be
deleted entirely.
Definition~\ref{def:group} is applied to the output of the infinite input $X_1,X_2,\dots$, in which
every output run is followed by a run of the opposite symbol with a surviving bit with probability one;
in a finite window, the group of the last output run need not be determined by the window.
The law of a typical output run used in Section~\ref{sec:dm} is the limit of the empirical distribution
of the groups and lengths of the output runs~\cite{DrineaMitzenmacher2007}; only the last output run of
a window can be affected by the window boundary, and the limit does not depend on it.
Drinea and Mitzenmacher write $F(i,z,r,s')$ for the set of groups of type $i$ whose first run has
length $z$, whose other same-symbol runs have total length $r$, and whose deleted opposite-symbol
runs have total length $s'$.

\subsection{Renewal weights}

Let $L'_1,L'_2,\dots$ be i.i.d.\ with law $P$.
For integers $r,i\ge0$ let
\begin{equation}
\label{eq:Q-Psi}
Q_{r,i}=\Pr\Bigl[\,\sum_{m=1}^{i}L'_m=r\Bigr],
\qquad
\Psi_r=\sum_{i\ge0}D^{\,i}Q_{r,i},
\end{equation}
where the empty sum is zero, so that $Q_{r,0}=\ones\{r=0\}$.
Since $L'_m\ge1$, $Q_{r,i}=0$ for $i>r$, and the series defining $\Psi_r$ is a finite sum.
We set $\Psi_r=0$ for $r<0$.
Equivalently, $\Psi_r=\E\bigl[\sum_{i\ge0}D^{\,i}\ones\{L'_1+\dots+L'_i=r\}\bigr]$ is the expected
number of partial sums equal to $r$, where the partial sum of $i$ terms is counted with weight
$D^{\,i}$; it appears in both functionals.

\begin{lemma}[Renewal weights]
\label{lem:psi}
Let $0<d<1$ and let $P$ have finite support.
\begin{enumerate}
\item The weights satisfy the renewal recursion
\begin{equation}
\label{eq:psi-rec}
\Psi_0=1,\qquad\Psi_r=D\sum_{\ell=1}^{\min(r,Z)}P_\ell\Psi_{r-\ell}\quad(r\ge1).
\end{equation}
Equivalently, for every integer $a\ge0$,
\begin{equation}
\label{eq:psi-conv}
\sum_{z=1}^{Z}P_z\Psi_{a-z}=\frac{\Psi_a-\ones\{a=0\}}{D}.
\end{equation}
\item The following series converge, with
\begin{align}
\label{eq:psi-sum}
\sum_{r\ge0}\Psi_r&=\frac{1}{1-D},&
\sum_{r\ge1}r\Psi_r&=\frac{\mu D}{(1-D)^2},\\
\label{eq:psi-dsum}
\sum_{r\ge0}\Psi_rd^{\,r}&=\frac{1}{1-D^2}.
\end{align}
\item For $0\le x\le1$, $\sum_{r\ge0}\Psi_rx^r=1/(1-D\hat P(x))$ with $\hat P(x)=\sum_zP_zx^z$.
\end{enumerate}
\end{lemma}

\begin{proof}
Conditioning on $L'_1$ gives $Q_{r,i}=\sum_{\ell}P_\ell Q_{r-\ell,i-1}$ for $i\ge1$, with
$Q_{r',i-1}=0$ for $r'<0$.
For $r\ge1$ we have $Q_{r,0}=0$, and therefore
\begin{align*}
\Psi_r&=\sum_{i\ge1}D^{\,i}\sum_{\ell}P_\ell Q_{r-\ell,i-1}\\
&=D\sum_{\ell}P_\ell\sum_{i'\ge0}D^{\,i'}Q_{r-\ell,i'}
=D\sum_\ell P_\ell\Psi_{r-\ell},
\end{align*}
where all sums are finite.
Also $\Psi_0=Q_{0,0}=1$.
Since $D>0$, dividing by $D$ gives~\eqref{eq:psi-conv} for $a\ge1$; for $a=0$ both sides
of~\eqref{eq:psi-conv} vanish.

All terms below are nonnegative, and the order of summation may be exchanged.
Since $\sum_rQ_{r,i}=1$, $\sum_rrQ_{r,i}=\E[L'_1+\dots+L'_i]=i\mu$, and
$\sum_rd^{\,r}Q_{r,i}=\E[d^{\,L'_1+\dots+L'_i}]=D^{\,i}$ by independence,
\begin{align*}
\sum_{r}\Psi_r&=\sum_{i\ge0}D^{\,i}=\frac{1}{1-D},\\
\sum_rr\Psi_r&=\mu\sum_{i\ge0}iD^{\,i}=\frac{\mu D}{(1-D)^2},\\
\sum_r\Psi_rd^{\,r}&=\sum_{i\ge0}D^{\,2i}=\frac{1}{1-D^2}.
\end{align*}
The same computation with $x^r$ in place of $d^{\,r}$ gives
$\sum_r\Psi_rx^r=\sum_i(D\hat P(x))^i$, and $0\le D\hat P(x)\le D<1$.
\end{proof}

%% file: sec_dm_TIT.tex
\section{The Drinea--Mitzenmacher Functional}
\label{sec:dm}

Let $K$ be the length of an output run and $T$ its group (Definition~\ref{def:group}), both under
the law of a typical output run of the renewal input, that is, the limit of the empirical
distribution of the lengths and groups of the output runs of a long block, as
in~\cite{DrineaMitzenmacher2007}.
Summing the Bernoulli deletions over the runs of a group gives~\cite[Eq.~(7)]{DrineaMitzenmacher2007}
\begin{multline}
\label{eq:dm-joint}
\Pr\bigl[T\in F(i,z,r,s'),\,K=k\bigr]
=P_zQ_{r,i}Q_{s',i}\,d^{\,z+r+s'}
\\
\quad\times\Bigl(\frac{\db}{d}\Bigr)^{k}
\Bigl[\tbinom{z+r}{k}-\tbinom{r}{k}\Bigr],
\end{multline}
for $i,r,s'\ge0$ and $z,k\ge1$.
The bracket counts the sets of $k$ surviving positions among the $z+r$ same-symbol bits of the
group that contain at least one position of the first run, and $d^{\,z+r+s'}(\db/d)^k$ is the
probability of each such survival pattern together with the deletion of the $s'$ opposite-symbol
bits.
Let $q_k=\Pr[K=k]$ be the law of the output run length.

\begin{theorem}[{Drinea--Mitzenmacher~\cite[Thm.~3]{DrineaMitzenmacher2007}}]
\label{thm:dm}
If $P$ is geometric or has finite support, then
\begin{equation}
\label{eq:dm-lb}
\Cdel\ge\RDM(P,d):=-h(d)+\frac{\db}{\E[K]}\Bigl(\HH(q)+\Lambda(P)\Bigr),
\end{equation}
where
\begin{multline}
\label{eq:dm-lambda}
\Lambda(P)=\sum_{k\ge1}\sum_{i,z,r,s'}
\Pr\bigl[T\in F(i,z,r,s'),\,K=k\bigr]
\\
\quad\times\log\Bigl[\tbinom{z+r}{k}-\tbinom{r}{k}\Bigr].
\end{multline}
\end{theorem}

The factor $\db/\E[K]$ is the number of output runs per input bit.
The term $\Lambda(P)$ is the expected logarithm of the number of survival patterns consistent with an
output run and its group, and it measures how much the jigsaw decoder of~\cite{DrineaMitzenmacher2007}
gains from knowing the group structure.
Drinea and Mitzenmacher evaluated~\eqref{eq:dm-lb} for geometric laws and for a two-parameter family,
and noted that it is not clear how the law should be chosen~\cite{DrineaMitzenmacher2007}.
Rubinstein and Con~\cite{RubinsteinCon2023} optimized a surrogate of~\eqref{eq:dm-lb} over general
laws and substituted the result into~\eqref{eq:dm-lb}.
The next lemma removes the need for a surrogate: it reduces every quantity in~\eqref{eq:dm-lb} to
series in $P$ and its renewal weights, each of which is itself a finite sum.
The series are infinite even for finite-support $P$ (for example, $P_1=1$ gives $\Psi_r=d^{\,r}$ for
all $r$), and Proposition~\ref{prop:dm-trunc} below shows that every finite truncation of them gives a
valid lower bound.

\subsection{Renewal reduction}

For $z\ge1$ and $r,k\ge0$ let
\begin{align}
\label{eq:w}
w(z,r,k)&=\Bigl[\tbinom{z+r}{k}-\tbinom{r}{k}\Bigr]\db^{\,k}d^{\,z+r-k},\\
\label{eq:G}
G_{z,r}&=\sum_{k=1}^{z+r}w(z,r,k)\log\Bigl[\tbinom{z+r}{k}-\tbinom{r}{k}\Bigr].
\end{align}
For $1\le k\le z+r$, the bracket in~\eqref{eq:w} is the number of $k$-subsets of $z+r$ positions
that meet the first $z$ positions; it is at least $1$ because $z\ge1$.
For $k>z+r$ both binomial coefficients vanish, $w(z,r,k)=0$, and such $k$ are excluded from all sums that carry the logarithm of the bracket.
Hence $w(z,r,k)\ge0$, the logarithm in~\eqref{eq:G} is nonnegative whenever $w(z,r,k)>0$, and
$0\le G_{z,r}\le(z+r)\sum_kw(z,r,k)\le z+r$.

\begin{lemma}[Renewal reduction]
\label{lem:dm-reduce}
Let $0<d<1$ and let $P$ have finite support.
Then, for every $k\ge1$,
\begin{align}
\label{eq:q-bilinear}
q_k&=\sum_{z\ge1}\sum_{r\ge0}P_z\Psi_r\,w(z,r,k)
=\frac{1-D^2}{D}\sum_{a\ge1}\Psi_a\,\Bin_a(k),\\
\label{eq:lambda-bilinear}
\Lambda(P)&=\sum_{z\ge1}\sum_{r\ge0}P_z\Psi_r\,G_{z,r},\\
\label{eq:EK}
\E[K]&=\db\,\mu\,\frac{1+D}{1-D}.
\end{align}
All series converge, and $\sum_{k\ge1}q_k=1$.
In particular, the number of output runs per input bit is
\begin{equation}
\label{eq:rho}
\rho=\frac{\db}{\E[K]}=\frac{1-D}{\mu(1+D)}.
\end{equation}
\end{lemma}

\begin{proof}
\emph{Step 1: summation over $s'$ and $i$.}
All terms of~\eqref{eq:dm-joint} are nonnegative, and sums may be taken in any order.
For fixed $(i,z,r,k)$, the dependence of~\eqref{eq:dm-joint} on $s'$ is through the factor
$Q_{s',i}d^{\,s'}$, and
\[
\sum_{s'\ge0}Q_{s',i}d^{\,s'}=\E\bigl[d^{\,L'_1+\dots+L'_i}\bigr]=D^{\,i}
\]
by independence (the probability that $i$ runs are deleted entirely).
The remaining factor is $P_zQ_{r,i}\,d^{\,z+r}(\db/d)^k[\binom{z+r}{k}-\binom{r}{k}]=P_zQ_{r,i}\,w(z,r,k)$.
Summing over $i$ gives $\sum_iD^{\,i}Q_{r,i}=\Psi_r$.
Hence
\begin{align*}
q_k&=\sum_{i,z,r,s'}\Pr\bigl[T\in F(i,z,r,s'),K=k\bigr]\\
&=\sum_{z\ge1}\sum_{r\ge0}P_z\Psi_rw(z,r,k),
\end{align*}
which is the first equality in~\eqref{eq:q-bilinear}.
The same computation with the factor $\log[\binom{z+r}{k}-\binom{r}{k}]$, which does not depend on
$(i,s')$, followed by summation over $k$, gives~\eqref{eq:lambda-bilinear}.
The series in~\eqref{eq:lambda-bilinear} converges because
$\sum_{z,r}P_z\Psi_rG_{z,r}\le\sum_{z,r}P_z\Psi_r(z+r)\le\mu/(1-D)+\mu D/(1-D)^2$
by Lemma~\ref{lem:psi}.

\emph{Step 2: the binomial mixture.}
Since $\binom{r}{k}\db^{\,k}d^{\,z+r-k}=d^{\,z}\Bin_r(k)$,
\begin{equation}
\label{eq:w-split}
w(z,r,k)=\Bin_{z+r}(k)-d^{\,z}\Bin_r(k).
\end{equation}
Substituting $a=z+r$ and using~\eqref{eq:psi-conv},
\begin{align*}
\sum_{z\ge1}\sum_{r\ge0}P_z\Psi_r\Bin_{z+r}(k)
&=\sum_{a\ge1}\Bin_a(k)\sum_{z}P_z\Psi_{a-z}\\
&=\frac1D\sum_{a\ge1}\Psi_a\Bin_a(k),
\end{align*}
where $a\ge1$ because $z\ge1$.
Summing first over $z$,
\begin{align*}
\sum_{z\ge1}\sum_{r\ge0}P_zd^{\,z}\Psi_r\Bin_r(k)&=D\sum_{r\ge0}\Psi_r\Bin_r(k)\\
&=D\sum_{a\ge1}\Psi_a\Bin_a(k),
\end{align*}
where the term $r=0$ vanishes because $\Bin_0(k)=0$ for $k\ge1$.
Both series are bounded by $\sum_a\Psi_a<\infty$.
Subtracting the second from the first gives the second equality in~\eqref{eq:q-bilinear}, since
$1/D-D=(1-D^2)/D$.

\emph{Step 3: normalization and mean.}
By~\eqref{eq:q-bilinear}, the identity $\sum_{k\ge1}\Bin_a(k)=1-d^{\,a}$, and Lemma~\ref{lem:psi},
\begin{align*}
\sum_{k\ge1}q_k&=\frac{1-D^2}{D}\sum_{a\ge1}\Psi_a(1-d^{\,a})\\
&=\frac{1-D^2}{D}\Bigl[\Bigl(\frac{1}{1-D}-1\Bigr)-\Bigl(\frac{1}{1-D^2}-1\Bigr)\Bigr]\\
&=\frac{1-D^2}{D}\Bigl[\frac{D}{1-D}-\frac{D^2}{1-D^2}\Bigr]=(1+D)-D=1.
\end{align*}
Similarly, $\sum_kk\Bin_a(k)=a\db$ and~\eqref{eq:psi-sum} give
\begin{align*}
\E[K]&=\frac{1-D^2}{D}\,\db\sum_{a\ge1}a\Psi_a\\
&=\frac{1-D^2}{D}\,\db\,\frac{\mu D}{(1-D)^2}
=\db\,\mu\,\frac{1+D}{1-D},
\end{align*}
which is~\eqref{eq:EK}; \eqref{eq:rho} follows.
\end{proof}

By~\eqref{eq:q-bilinear}, $q$ is the law of a $\mathrm{Bin}(a,\db)$ variable conditioned to be
positive, where $a$ is drawn with probability proportional to $\Psi_a(1-d^{\,a})$.
This is a representation of $q$, not a description of the group parse.
The rate~\eqref{eq:rho} has a direct interpretation, which serves as a check: by
Lemma~\ref{lem:residual} below, an output bit starts a new output run with probability
$(1-D)/(\mu\db(1+D))$, and there are $\db$ output bits per input bit.

The dependence of~\eqref{eq:dm-lb} on $P$ enters through $\mu$, $D$, the renewal weights, and the
bilinear forms in $(P,\Psi)$.
The kernel $G$ and the binomial laws $\Bin_a$ depend on $d$ alone.
Once they are tabulated, each evaluation of a truncation of $\RDM(P,d)$ costs one renewal recursion,
one bilinear form $P^{\top}G\Psi$, and one matrix--vector product for $q$, and the exact gradient of
this truncated expression with respect to $P$ is available by differentiating these operations.

\subsection{One-sided truncation}

Numerical evaluation requires finite sums.
By Step~1 of the proof of Lemma~\ref{lem:dm-reduce} and the remarks after~\eqref{eq:G}, every term
of the series~\eqref{eq:q-bilinear} and~\eqref{eq:lambda-bilinear} is nonnegative.

\begin{proposition}[One-sided truncation]
\label{prop:dm-trunc}
Let $P$ have finite support, let $\mathcal A\subset\{1,2,\dots\}$ and $\mathcal K\subset\{1,2,\dots\}$ be
finite, and let $\mathcal T$ be a finite set of triples $(z,r,k)$ with $z\ge1$, $r\ge0$, and $1\le k\le z+r$.
Define
\begin{align*}
\tilde q_k&=\frac{1-D^2}{D}\sum_{a\in\mathcal A}\Psi_a\Bin_a(k),\quad k\in\mathcal K,\\
\tilde\Lambda&=\sum_{(z,r,k)\in\mathcal T}P_z\Psi_r\,w(z,r,k)\log\Bigl[\tbinom{z+r}{k}-\tbinom{r}{k}\Bigr],
\end{align*}
and $\varepsilon=1-\sum_{k\in\mathcal K}\tilde q_k$.
Then $0\le\varepsilon\le1$ and
\begin{multline}
\label{eq:dm-trunc}
\widetilde{\mathcal R}:=-h(d)+\frac{\db}{\E[K]}
\Bigl(-\sum_{k\in\mathcal K}\tilde q_k\log\tilde q_k
\\
\quad+(1-\varepsilon)\log(1-\varepsilon)+\tilde\Lambda\Bigr)
\\
\le\RDM(P,d)\le\Cdel ,
\end{multline}
where $\E[K]$ is given by~\eqref{eq:EK} and $0\log0=0$.
\end{proposition}

\begin{proof}
Every omitted term is nonnegative; hence $0\le\tilde q_k\le q_k$ for $k\in\mathcal K$ and
$0\le\tilde\Lambda\le\Lambda(P)$.
In particular $\sum_{k\in\mathcal K}\tilde q_k\le\sum_kq_k=1$, and $0\le\varepsilon\le1$.

Extend $\tilde q$ by $\tilde q_k=0$ for $k\notin\mathcal K$.
We show
\begin{equation}
\label{eq:Hq-lower}
\HH(q)\ge-\sum_{k}\tilde q_k\log\tilde q_k+(1-\varepsilon)\log(1-\varepsilon).
\end{equation}
If $\varepsilon=1$, then $\tilde q=0$, the right-hand side is zero, and~\eqref{eq:Hq-lower} holds.
If $\varepsilon=0$, then $\tilde q\le q$ and both sum to one, so $\tilde q=q$ and~\eqref{eq:Hq-lower}
holds with equality.
Let $0<\varepsilon<1$, and define the probability vectors $\bar q=\tilde q/(1-\varepsilon)$ and
$\bar e=(q-\tilde q)/\varepsilon$, so that $q=(1-\varepsilon)\bar q+\varepsilon\bar e$.
The function $\phi(x)=-x\log x$ is concave on $[0,1]$; hence, for every $k$,
\begin{align*}
\phi(q_k)&\ge(1-\varepsilon)\phi(\bar q_k)+\varepsilon\phi(\bar e_k)\\
&\ge(1-\varepsilon)\phi(\bar q_k).
\end{align*}
Summing over $k$ gives $\HH(q)\ge(1-\varepsilon)\HH(\bar q)$, and
\begin{align*}
(1-\varepsilon)\HH(\bar q)&=-\sum_k\tilde q_k\log\frac{\tilde q_k}{1-\varepsilon}\\
&=-\sum_k\tilde q_k\log\tilde q_k+(1-\varepsilon)\log(1-\varepsilon).
\end{align*}

The prefactor $\db/\E[K]$ is positive and is computed exactly by~\eqref{eq:EK}, without truncation.
Combining~\eqref{eq:Hq-lower} with $\tilde\Lambda\le\Lambda(P)$ gives
$\widetilde{\mathcal R}\le\RDM(P,d)$, and the last inequality of~\eqref{eq:dm-trunc} is
Theorem~\ref{thm:dm}.
\end{proof}

Proposition~\ref{prop:dm-trunc} allows the sets $\mathcal A$, $\mathcal K$, and $\mathcal T$ to be
chosen independently of each other; no tail estimate is needed for validity.
In our evaluations the omitted mass satisfies $\varepsilon<10^{-12}$.
Section~\ref{sec:numerics} describes how the finite sums are enclosed in interval arithmetic, so that
the rounding errors of the computation are also accounted for.

%% file: sec_vtr_TIT.tex
\section{The Venkataramanan--Tatikonda--Ramchandran Functional for General Run Laws}
\label{sec:vtr}

This section proves Theorem~\ref{thm:vtr} below, which extends the bound
of~\cite[Thm.~3]{VenkataramananTatikondaRamchandran2013} from Markov inputs to every finite-support
run-length law and replaces its correction term by a larger one.
The proof bounds the four terms of a decomposition of $\I(X^n;Y)$ separately:
Lemma~\ref{lem:runlevel} gives the decomposition, Lemmas~\ref{lem:residual} and~\ref{lem:hsy} and
Proposition~\ref{prop:filter} bound the second term, Lemma~\ref{lem:hxsy} bounds the third, and
Lemma~\ref{lem:chain} bounds the fourth.
Throughout, $X_1,X_2,\dots$ is the stationary renewal input of Definition~\ref{def:renewal} with a
finite-support law $P$, $0<d<1$, and the notation of Section~\ref{subsec:runs} is used.

\subsection{Decomposition}

Following~\cite{VenkataramananTatikondaRamchandran2013}, we introduce the numbers of runs that are
deleted entirely between consecutive output bits.
For $1\le t\le M_n$ let $J(t)$ be the index of the input run that contains output bit $t$, and set
$J(0)=0$ and $J(M_n+1)=R_n+1$.
For $t=1,\dots,M_n+1$ let
\begin{equation}
\label{eq:S-def}
S_t=\bigl|\{j:\ J(t-1)<j<J(t),\ \tilde N_j=0\}\bigr|,
\end{equation}
and $S=(S_1,\dots,S_{M_n+1})$.
Since no bit survives strictly between two consecutive output bits, every run $j$ with
$J(t-1)<j<J(t)$ has $\tilde N_j=0$; hence $S_t=J(t)-J(t-1)-1$ if $J(t)>J(t-1)$, and $S_t=0$ if
output bits $t-1$ and $t$ lie in the same run, that is, $S_t=\max\{J(t)-J(t-1)-1,0\}$.
Thus $S_1$ counts the runs before the first output bit, $S_{M_n+1}$ the runs after the last, and
$S_t$ for $2\le t\le M_n$ the runs deleted entirely between output bits $t-1$ and $t$.

\begin{lemma}[Run-level description]
\label{lem:runlevel}
\begin{enumerate}
\item The pair $(Y,S)$ determines $R_n$ and $(\tilde N_1,\dots,\tilde N_{R_n})$, and, if $M_n\ge1$,
the symbol $B$ of the first run.
\item Given $X^n$, the vector $(\tilde N_1,\dots,\tilde N_{R_n})$ determines $(Y,S)$.
\item The mutual information satisfies
\begin{multline}
\label{eq:vtr-split}
\I(X^n;Y)=\HH(X^n)-\HH(S\mid Y)
\\
\quad-\HH(X^n\mid S,Y)+\HH(S\mid X^n,Y).
\end{multline}
\item $\HH(S\mid X^n,Y)=\HH(\tilde N_1,\dots,\tilde N_{R_n}\mid X^n,Y)$.
\end{enumerate}
\end{lemma}

\begin{proof}
\emph{Part 1.}
If $M_n=0$, then $R_n=S_1$ and $\tilde N_j=0$ for all $j$.
Let $M_n\ge1$.
We have $J(1)=S_1+1$.
For $2\le t\le M_n$, if output bits $t-1$ and $t$ lie in the same run, then $S_t=0$ and
$Y_t=Y_{t-1}$.
If they lie in different runs, then $J(t)=J(t-1)+S_t+1$, and if moreover $S_t=0$, the runs $J(t-1)$
and $J(t)$ are adjacent, have opposite symbols, and $Y_t\ne Y_{t-1}$.
Hence
\[
J(t)=
\begin{cases}
J(t-1), & S_t=0\text{ and }Y_t=Y_{t-1},\\
J(t-1)+S_t+1, & \text{otherwise},
\end{cases}
\]
which determines $J(1),\dots,J(M_n)$ recursively from $(Y,S)$.
Then $R_n=J(M_n)+S_{M_n+1}$, $\tilde N_j=|\{t\le M_n:J(t)=j\}|$, and
$B=Y_1\oplus((J(1)-1)\bmod2)$.

\emph{Part 2.}
Given $X^n$, the run lengths $\ell_j$ and the symbols of all runs are known.
The output $Y$ is the concatenation of $\tilde N_j$ copies of the symbol of run $j$,
$j=1,\dots,R_n$, and $S$ is obtained from~\eqref{eq:S-def}.

\emph{Part 3.}
By the chain rule,
$\HH(X^n,S\mid Y)=\HH(S\mid Y)+\HH(X^n\mid S,Y)=\HH(X^n\mid Y)+\HH(S\mid X^n,Y)$, and all entropies are
finite because all variables take finitely many values.
Substituting $\HH(X^n\mid Y)$ into $\I(X^n;Y)=\HH(X^n)-\HH(X^n\mid Y)$ gives~\eqref{eq:vtr-split}.

\emph{Part 4.}
Write $\tilde N=(\tilde N_1,\dots,\tilde N_{R_n})$.
By Part~1, $\tilde N$ is a function of $(Y,S)$; by Part~2, $S$ is a function of $(X^n,\tilde N)$.
Hence $\HH(S\mid X^n,Y)=\HH(S,\tilde N\mid X^n,Y)=\HH(\tilde N\mid X^n,Y)$.
\end{proof}

The first term of~\eqref{eq:vtr-split} is controlled by~\eqref{eq:entropy-rate}.
VTR bounded the second and third terms for Markov inputs, for which the output $Y$ is itself a
first-order Markov chain, and lower-bounded the fourth term by counting output runs formed from
three adjacent input runs~\cite[Thm.~3]{VenkataramananTatikondaRamchandran2013}.
For a general run law the output is not Markov, and each term requires a separate treatment.

\subsection{The output process and the residual chain}
\label{subsec:residual}

To study $\HH(S\mid Y)$ we pass to the infinite input $X_1,X_2,\dots$ and the deletion indicators
$\Delta_1,\Delta_2,\dots$.
Let $\tau_1<\tau_2<\cdots$ be the positions $k$ with $\Delta_k=0$; there are infinitely many
with probability one.
For $t\ge1$ let $\bar Y_t=X_{\tau_t}$, let $\bar J_t$ be the index of the run that contains position
$\tau_t$, and let
\[
U_t=\theta_{\bar J_t}-\tau_t\in\{0,\dots,Z-1\}
\]
be the number of bits of that run after position $\tau_t$.
Let $\bar S_1=\bar J_1-1$ and $\bar S_t=\max\{\bar J_t-\bar J_{t-1}-1,0\}$ for $t\ge2$, the number of
runs strictly between $\bar J_{t-1}$ and $\bar J_t$ (zero if $\bar J_t=\bar J_{t-1}$); every such run is
deleted entirely.
The output of the window is a prefix of $\bar Y$: $Y=\bar Y^{M_n}$, and $J(t)=\bar J_t$ for $t\le M_n$.
Moreover, $S_t=\bar S_t$ for $1\le t\le M_n$, because the runs counted by $S_t$ precede run $J(t)$
and therefore lie entirely inside the window.
For $u\ge0$ let $\zeta_u=1-d^{\,u}+d^{\,u}D/(1+D)$, and let
\begin{equation}
\label{eq:nu}
\nu_v=\frac{\db}{1-D}\sum_{\ell=v+1}^{Z}P_\ell\,d^{\,\ell-1-v},\qquad v=0,\dots,Z-1.
\end{equation}
Since $\sum_{v<\ell}\db\,d^{\,\ell-1-v}=1-d^{\,\ell}$, $\sum_v\nu_v=\sum_\ell P_\ell(1-d^{\,\ell})/(1-D)=1$.

\begin{lemma}[Output process]
\label{lem:residual}
\begin{enumerate}
\item The process $(\bar Y_t,U_t,\bar S_{t+1})_{t\ge1}$ is stationary.
For every $t\ge1$, $U_t$ has the law $\pi$ of~\eqref{eq:pi}, $\bar Y_t$ is uniform, and $U_t$ and
$\bar Y_t$ are independent.
\item The process $(U_t,\bar Y_t)_{t\ge1}$ is a time-homogeneous Markov chain on
$\{0,\dots,Z-1\}\times\{0,1\}$.
More precisely, for $t\ge2$, conditionally on $(U^{t-1},\bar Y^{t-1},\bar S^{t-1})$ with $U_{t-1}=u$,
the triple $(\bar S_t,U_t,\bar Y_t)$ has the following law:
\begin{enumerate}
\item for $0\le v<u$, with probability $\db\,d^{\,u-1-v}$: $\bar S_t=0$, $\bar Y_t=\bar Y_{t-1}$, and
$U_t=v$;
\item with the remaining probability $d^{\,u}$: $\bar S_t$ is geometric,
$\Pr[\bar S_t=m]=(1-D)D^{\,m}$ for $m\ge0$; $\bar Y_t\ne\bar Y_{t-1}$ if and only if $\bar S_t$ is
even; and $U_t$ is independent of $\bar S_t$ with law $\nu$.
\end{enumerate}
\item Let $\alpha=\sum_u\pi(u)d^{\,u}$.
Then $\alpha=(1-D)/(\mu\db)$ and $\Pr[\bar Y_t\ne\bar Y_{t-1}]=\alpha/(1+D)$ for $t\ge2$.
\end{enumerate}
\end{lemma}

\begin{proof}
\emph{Part 1.}
For $t\ge1$ let $W_t=\bigl((X_{\tau_t+k})_{k\ge0},(\Delta_{\tau_t+k})_{k\ge1}\bigr)$.
Fix $t\ge1$ and $p\ge1$.
The event $\{\tau_t=p\}$ is determined by $\Delta_1,\dots,\Delta_p$, which are independent of
$(X_k)_{k\ge1}$ and of $(\Delta_{p+k})_{k\ge1}$.
Hence, conditionally on $\tau_t=p$, the sequence $(X_{p+k})_{k\ge0}$ has the law of
$(X_{1+k})_{k\ge0}$ by Lemma~\ref{lem:stationary}, the sequence $(\Delta_{p+k})_{k\ge1}$ is i.i.d.\
Bernoulli$(d)$, and the two are independent.
This conditional law does not depend on $p$ or on $t$; hence $W_t$ has the same law for every $t$.
For $s\ge t$, the variables $\bar Y_s$, $U_s$, and $\bar S_{s+1}$ are obtained from $W_t$ by the same
rule for every $t$: $\bar Y_t$ is the first symbol of $W_t$, $U_t$ is the number of consecutive symbols
after it that are equal to it, later survivors are read off the deletion indicators, and run
boundaries are the positions where the symbol changes.
Hence $(\bar Y_s,U_s,\bar S_{s+1})_{s\ge t}$ has the same law for every $t$, which is stationarity.
For $t=1$, conditionally on $\tau_1=p$, the pair $(\bar Y_1,U_1)=(X_p,V_p)$ has the law of
$(X_1,V_1)$, which is described by Lemma~\ref{lem:stationary}.

\emph{Part 2.}
Fix $t\ge2$ and $p\ge1$, let $r(p)$ be the index of the run containing position $p$, and let
$\mathcal F_p$ be the $\sigma$-field of events $A$ such that
$A\cap\{r(p)=r\}\in\sigma(B,L_1,\dots,L_r,\Delta_1,\dots,\Delta_p)$ for every $r$.
For every $r$, the event $\{r(p)=r\}$ is determined by $L_1,\dots,L_r$, and the vector
$(L_{r+1},L_{r+2},\dots)$ is independent of $\sigma(B,L_1,\dots,L_r,\Delta_1,\dots,\Delta_p)$; hence,
conditionally on $\mathcal F_p$, the run lengths $L_{r(p)+1},L_{r(p)+2},\dots$ are i.i.d.\ with law $P$, the
indicators $\Delta_{p+1},\Delta_{p+2},\dots$ are i.i.d.\ Bernoulli$(d)$, and these two sequences are
independent.
The event $\{\tau_{t-1}=p\}$ and, on this event, the variables $U^{t-1}$, $\bar Y^{t-1}$, and
$\bar S^{t-1}$ are $\mathcal F_p$-measurable, and $U_{t-1}=\theta_{r(p)}-p$.
On $\{\tau_{t-1}=p,\ U_{t-1}=u\}$, the next surviving bit is the first $k>p$ with $\Delta_k=0$.
It lies in the current run and equals $p+k'$ with $1\le k'\le u$ with probability $d^{\,k'-1}\db$;
then $\bar S_t=0$, $\bar Y_t=\bar Y_{t-1}$, and $U_t=u-k'$.
Writing $v=u-k'$ gives case~(a).
With probability $d^{\,u}$ the remaining $u$ bits of the current run are deleted.
Then the runs $r(p)+1,r(p)+2,\dots$ are examined in order; each is deleted entirely with probability
$D$, independently, and $\bar S_t$ is the number of runs deleted before the first run with a
surviving bit.
Hence $\Pr[\bar S_t=m]=(1-D)D^{\,m}$.
The run $\bar J_t=r(p)+\bar S_t+1$ has the symbol of run $r(p)$ flipped $\bar S_t+1$ times, which
gives the parity rule.
Conditionally on $\bar S_t=m$, the run $\bar J_t$ is independent of the $m$ deleted runs before it
and has the law of a run conditioned to have a surviving bit: its length is $\ell$ with probability
$P_\ell(1-d^{\,\ell})/(1-D)$, and its first surviving bit is its $k''$-th bit with probability
$d^{\,k''-1}\db/(1-d^{\,\ell})$, in which case $U_t=\ell-k''$.
Hence $\Pr[U_t=v\mid\bar S_t=m]=\sum_{\ell>v}P_\ell d^{\,\ell-1-v}\db/(1-D)=\nu_v$, for every $m$,
which gives case~(b).
The resulting conditional law depends on $(U^{t-1},\bar Y^{t-1},\bar S^{t-1})$ only through
$(U_{t-1},\bar Y_{t-1})$ and does not depend on $p$ or $t$.

\emph{Part 3.}
Using $\Pr[L>u]=\sum_{\ell>u}P_\ell$ and exchanging the sums,
\[
\alpha=\frac1\mu\sum_{\ell}P_\ell\sum_{u=0}^{\ell-1}d^{\,u}
=\frac1\mu\sum_\ell P_\ell\frac{1-d^{\,\ell}}{1-d}=\frac{1-D}{\mu\db}.
\]
By Part~2, a flip requires case~(b) with $\bar S_t$ even, and
$\Pr[\bar S_t\text{ even}]=\sum_{j\ge0}(1-D)D^{\,2j}=1/(1+D)$.
By Part~1, $U_{t-1}\sim\pi$; hence $\Pr[\bar Y_t\ne\bar Y_{t-1}]=\E[d^{\,U_{t-1}}]/(1+D)=\alpha/(1+D)$.
\end{proof}

Part~1 implies that $\pi$ is an invariant law of the chain of Part~2.
This can also be checked directly: for $0\le v\le Z-1$, using $\alpha\db/(1-D)=1/\mu$,
\begin{align*}
&\sum_{u>v}\pi(u)\db\,d^{\,u-1-v}+\alpha\,\nu_v\\
&\quad=\frac1\mu\Bigl[\sum_{\ell>v+1}P_\ell\bigl(1-d^{\,\ell-1-v}\bigr)+\sum_{\ell>v}P_\ell d^{\,\ell-1-v}\Bigr]\\
&\quad=\frac{\Pr[L>v]}{\mu}.
\end{align*}

For $\kappa\ge1$ and $t\ge\kappa+1$, the vector $(\bar Y_{t-\kappa},\dots,\bar Y_t,\bar S_t)$ is a
function of $(\bar Y_s,U_s,\bar S_{s+1})_{s=t-\kappa}^{t}$, and by Part~1 of
Lemma~\ref{lem:residual} its law does not depend on $t$.
We write
\[
\HH_\pi(S_t\mid Y_{t-\kappa},\dots,Y_t)=\HH(\bar S_t\mid\bar Y_{t-\kappa},\dots,\bar Y_t),\qquad t\ge\kappa+1.
\]
The next proposition shows that this quantity is a finite computation.
Let
\begin{equation}
\label{eq:g}
g=\frac{h(D^2)}{1-D^2}
\end{equation}
be the entropy of the geometric law $(1-D^2)D^{\,2j}$, $j\ge0$.
Define the $Z\times Z$ matrices, indexed by $u,v\in\{0,\dots,Z-1\}$,
\begin{align*}
T^{\mathrm{stay}}_{uv}&=\ones\{v<u\}\,\db\,d^{\,u-1-v},\\
T^{\mathrm{same}}_{uv}&=T^{\mathrm{stay}}_{uv}+d^{\,u}\frac{D}{1+D}\nu_v,\qquad
T^{\mathrm{flip}}_{uv}=d^{\,u}\frac{1}{1+D}\nu_v .
\end{align*}

\begin{proposition}[Evaluation of the context entropy]
\label{prop:filter}
Let $\kappa\ge1$.
For $\epsilon=(\epsilon_1,\dots,\epsilon_\kappa)\in\{\mathrm{same},\mathrm{flip}\}^\kappa$ let
$x_0=\pi$ and $x_j=x_{j-1}T^{\epsilon_j}$ (row vectors), let $\Pr[\epsilon]=x_\kappa\ones$, and, if
$x_{\kappa-1}\ones>0$, let $y=x_{\kappa-1}/(x_{\kappa-1}\ones)$ and
\begin{equation}
\label{eq:pstay}
p_{\mathrm{stay}}(\epsilon)=\frac{\sum_uy_u(1-d^{\,u})}{\sum_uy_u\zeta_u}.
\end{equation}
Then
\begin{equation}
\label{eq:hsy-filter}
\HH_\pi(S_t\mid Y_{t-\kappa},\dots,Y_t)=\sum_{\epsilon:\,\Pr[\epsilon]>0}\Pr[\epsilon]\,H_\epsilon,
\end{equation}
where $H_\epsilon=g$ if $\epsilon_\kappa=\mathrm{flip}$ and
$H_\epsilon=h(p_{\mathrm{stay}}(\epsilon))+(1-p_{\mathrm{stay}}(\epsilon))g$ if
$\epsilon_\kappa=\mathrm{same}$.
For $\kappa=1$, with $p_f=\alpha/(1+D)$ and $p_0=(1-\alpha)/(1-p_f)$,
\begin{equation}
\label{eq:hsy-pair}
\HH_\pi(S_t\mid Y_{t-1},Y_t)=(1-p_f)\bigl[h(p_0)+(1-p_0)g\bigr]+p_fg .
\end{equation}
\end{proposition}

\begin{proof}
Fix $t\ge\kappa+1$ and let $E_s=\ones\{\bar Y_s\ne\bar Y_{s-1}\}$ for $t-\kappa<s\le t$; we identify
$E_s=0$ with ``same'' and $E_s=1$ with ``flip''.

\emph{Reduction to polarity events.}
The bits $\bar Y_{t-\kappa},\dots,\bar Y_t$ and the pair $(\bar Y_{t-\kappa},E_{t-\kappa+1}^{t})$
determine each other.
By Part~2 of Lemma~\ref{lem:residual}, the conditional law of $(\bar S_s,U_s,E_s)$ given the past
depends only on $U_{s-1}$ and not on $\bar Y_{s-1}$.
Since $U_{t-\kappa}$ is independent of $\bar Y_{t-\kappa}$ by Part~1, the vector
$(U_{t-\kappa},(\bar S_s,U_s,E_s)_{s=t-\kappa+1}^{t})$ is independent of $\bar Y_{t-\kappa}$.
Hence
\[
\HH(\bar S_t\mid\bar Y_{t-\kappa}^t)=\HH(\bar S_t\mid\bar Y_{t-\kappa},E_{t-\kappa+1}^t)
=\HH(\bar S_t\mid E_{t-\kappa+1}^t).
\]

\emph{Law of the events.}
By Part~2 of Lemma~\ref{lem:residual}, for $u,v\in\{0,\dots,Z-1\}$,
$\Pr[E_s=0,U_s=v\mid U_{s-1}=u]$ is the sum of the probability $T^{\mathrm{stay}}_{uv}$ of case~(a)
and the probability $d^{\,u}\Pr[\bar S_s\text{ odd}]\nu_v=d^{\,u}\frac{D}{1+D}\nu_v$ of case~(b) with
$\bar S_s$ odd; this is $T^{\mathrm{same}}_{uv}$.
Similarly $\Pr[E_s=1,U_s=v\mid U_{s-1}=u]=T^{\mathrm{flip}}_{uv}$.
By the Markov property and $U_{t-\kappa}\sim\pi$,
\[
\Pr\bigl[E_{t-\kappa+j}=\epsilon_j\ (1\le j\le\kappa'),\ U_{t-\kappa+\kappa'}=v\bigr]=(x_{\kappa'})_v
\]
for $0\le\kappa'\le\kappa$.
With $\kappa'=\kappa$ this gives $\Pr[\epsilon]=x_\kappa\ones$; with $\kappa'=\kappa-1$ it shows that
$y$ is the conditional law of $U_{t-1}$ given the first $\kappa-1$ events.

\emph{Law of $\bar S_t$.}
Condition on the first $\kappa-1$ events and on $U_{t-1}=u$; by the Markov property, the pair
$(\bar S_t,E_t)$ has the law of Part~2 of Lemma~\ref{lem:residual}.
If $E_t=1$, then $\bar S_t$ is even and
$\Pr[\bar S_t=2j,E_t=1\mid U_{t-1}=u]=d^{\,u}(1-D)D^{\,2j}$, which is proportional to $D^{\,2j}$ with
a factor that does not depend on $j$.
Hence, given $\epsilon$ with $\epsilon_\kappa=\mathrm{flip}$, $\bar S_t/2$ has the geometric law
$(1-D^2)D^{\,2j}$ whatever the value of $U_{t-1}$, and $H_\epsilon=g$.
If $E_t=0$, then
\begin{align*}
\Pr[\bar S_t=0,E_t=0\mid U_{t-1}=u]&=1-d^{\,u},\\
\Pr[\bar S_t=2j+1,E_t=0\mid U_{t-1}=u]&=d^{\,u}(1-D)D^{\,2j+1},
\end{align*}
and $\Pr[E_t=0\mid U_{t-1}=u]=\zeta_u$.
Averaging over $U_{t-1}\sim y$ gives $\Pr[\bar S_t=0\mid\epsilon]=p_{\mathrm{stay}}(\epsilon)$, and,
given $\epsilon$ and $\bar S_t\ne0$, $(\bar S_t-1)/2$ again has the geometric law $(1-D^2)D^{\,2j}$.
The conditional law of $\bar S_t$ is therefore the mixture, with weights $p_{\mathrm{stay}}$ and
$1-p_{\mathrm{stay}}$, of the point mass at $0$ and a law on the odd integers with entropy $g$.
Since the two components have disjoint supports, its entropy is
$h(p_{\mathrm{stay}})+(1-p_{\mathrm{stay}})g$.
Averaging over $\epsilon$ gives~\eqref{eq:hsy-filter}.

Here $y=\pi$, $\sum_u\pi(u)(1-d^{\,u})=1-\alpha$, and
$\sum_u\pi(u)\zeta_u=1-\alpha+\alpha D/(1+D)=1-p_f$; hence $p_{\mathrm{stay}}=p_0$,
$\Pr[\mathrm{same}]=1-p_f$, and $\Pr[\mathrm{flip}]=p_f$.
\end{proof}

Since conditioning does not increase entropy, $\HH_\pi(S_t\mid Y_{t-\kappa},\dots,Y_t)$ is
nonincreasing in $\kappa$, and the bound of Theorem~\ref{thm:vtr} below is nondecreasing in $\kappa$.
For the untruncated geometric law $P_z=(1-p)p^{z-1}$ (which has infinite support and is outside the
scope of the lemmas; the matrices are then infinite), the same formulas give
$\pi(u)=(1-p)p^u$, $\nu=\pi$, and $\pi T^{\mathrm{stay}}\propto\pi$; every $x_j$ is then proportional to
$\pi$, and the filter gives no gain from $\kappa>1$.
For finite-support laws, including truncated geometric laws, a larger $\kappa$ can strictly decrease
the right-hand side of~\eqref{eq:hsy-filter}.

\begin{lemma}[Bound on $\HH(S\mid Y)$]
\label{lem:hsy}
For every $\kappa\ge1$,
\begin{equation}
\label{eq:hsy-bound}
\limsup_{n\to\infty}\frac1n\HH(S\mid Y)\le\db\,\HH_\pi(S_t\mid Y_{t-\kappa},\dots,Y_t).
\end{equation}
\end{lemma}

\begin{proof}
Write $M=M_n$, $H^\star=\HH_\pi(S_t\mid Y_{t-\kappa},\dots,Y_t)$, and $m_n=\lceil\db n+n^{2/3}\rceil$.

\emph{Step 1: the last component.}
Since $Y=\bar Y^{M}$ and $S=(\bar S^{M},S_{M+1})$ with $0\le S_{M+1}\le n$,
\begin{align*}
\HH(S\mid Y)&\le\HH(S_{M+1})+\HH(\bar S^{M}\mid\bar Y^{M})\\
&\le\log(n+1)+\HH(\bar S^{M}\mid\bar Y^{M}).
\end{align*}
Here $\bar S^M$ has the random length $M$, which is a function of $\bar Y^M$.

\emph{Step 2: fixing the length.}
Let $\mathcal E=\{M\le m_n\}$.
Since $M$ is a sum of $n$ independent Bernoulli$(\db)$ variables, Hoeffding's
inequality~\cite{Hoeffding1963} gives $\Pr[\mathcal E^c]\le e^{-2n^{1/3}}$.
Using $\HH(A\mid C)\le\HH(A\mid C,\ones_{\mathcal E})+\HH(\ones_{\mathcal E})$,
\begin{multline*}
\HH(\bar S^{M}\mid\bar Y^{M})\le1+\Pr[\mathcal E]\,\HH(\bar S^{M}\mid\bar Y^{M},\mathcal E)
\\
\quad+\Pr[\mathcal E^c]\,\HH(\bar S^{M}\mid\bar Y^{M},\mathcal E^c).
\end{multline*}
The components of $\bar S^{M}$ count distinct runs of the window; hence they are nonnegative integers
with sum at most $n$.
Given $M$, there are at most $\binom{n+M}{M}\le2^{n+M}\le4^n$ such vectors, and the last term is at
most $2n\Pr[\mathcal E^c]$.
On $\mathcal E$, $\bar S^M$ is a function of $(\bar S^{m_n},M)$; hence
\begin{align*}
\Pr[\mathcal E]\,\HH(\bar S^{M}\mid\bar Y^{M},\mathcal E)&\le\Pr[\mathcal E]\,\HH(\bar S^{m_n}\mid\bar Y^{M},\mathcal E)\\
&\le\HH(\bar S^{m_n}\mid\bar Y^{M},\ones_{\mathcal E}).
\end{align*}
Using $\HH(A\mid C)\le\HH(A\mid C')+\HH(C'\mid C)$ with $C=(\bar Y^M,\ones_{\mathcal E})$ and $C'=\bar Y^{m_n}$,
\[
\HH(\bar S^{m_n}\mid\bar Y^{M},\ones_{\mathcal E})\le\HH(\bar S^{m_n}\mid\bar Y^{m_n})+\E[(m_n-M)^+],
\]
because, given $\bar Y^M$, the prefix $\bar Y^{m_n}$ is determined by at most $(m_n-M)^+$ further binary
symbols.
Since $M\le n$,
$\E[(m_n-M)^+]\le m_n-\db n+\E[(M-m_n)^+]\le n^{2/3}+1+n\Pr[\mathcal E^c]$.

\emph{Step 3: context.}
By the chain rule and because conditioning does not increase entropy,
\begin{align*}
\HH(\bar S^{m_n}\mid\bar Y^{m_n})&\le\sum_{t=1}^{m_n}\HH(\bar S_t\mid\bar Y^{m_n})\\
&\le\sum_{t=1}^{\kappa}\HH(\bar S_t)+\sum_{t=\kappa+1}^{m_n}\HH(\bar S_t\mid\bar Y_{t-\kappa}^t).
\end{align*}
Each term of the second sum equals $H^\star$.
For the first sum, $\Pr[\bar S_t\ge j]\le D^{\,j-1}$ for $j\ge1$ and every $t$ (for $t\ge2$ by
Part~2 of Lemma~\ref{lem:residual}; for $t=1$ because runs $2,3,\dots$ are deleted entirely
independently with probability $D$), and hence $\E[\bar S_t]\le a:=1/(1-D)$.
A nonnegative integer random variable with mean at most $a$ has entropy at most
$c_D=(a+1)\log(a+1)-a\log a$~\cite[Ch.~12]{CoverThomas2006}.
Hence $\HH(\bar S^{m_n}\mid\bar Y^{m_n})\le\kappa c_D+m_nH^\star$.

Collecting the bounds,
\[
\HH(S\mid Y)\le m_nH^\star+\log(n+1)+n^{2/3}+\kappa c_D+2+3n\,e^{-2n^{1/3}}.
\]
Dividing by $n$ and using $m_n/n\to\db$ gives~\eqref{eq:hsy-bound}.
\end{proof}

\subsection{The run channel and $\HH(X^n\mid S,Y)$}

\begin{lemma}[Bound on $\HH(X^n\mid S,Y)$]
\label{lem:hxsy}
Let $(L,N)$ have the law $\Pr[L=r,N=m]=P_r\Bin_r(m)$.
Then
\begin{equation}
\label{eq:hxsy}
\limsup_{n\to\infty}\frac1n\HH(X^n\mid S,Y)\le\frac{\HH(L\mid N)}{\mu}.
\end{equation}
Moreover, $\HH(L\mid N)=\HH(\Omega)-\HH(\omega)$ for the joint matrix $\Omega_{r,m}=P_r\Bin_r(m)$ and its
column marginal $\omega_m=\sum_r\Omega_{r,m}$, the law of the number of survivors of a run.
\end{lemma}

\begin{proof}
If $Z=1$, the input is determined by $B$ and $\HH(L\mid N)=0$; assume $Z\ge2$.
Let $k_n=\lceil n/\mu+n^{2/3}\rceil$.

\emph{Step 1: reduction to run lengths.}
By Part~1 of Lemma~\ref{lem:runlevel}, $(R_n,\tilde N_1,\dots,\tilde N_{R_n})$ is a function of
$(S,Y)$, and $X^n$ is a function of $(B,\ell_1,\dots,\ell_{R_n})$.
The lengths $\ell_1$ and $\ell_{R_n}$ take values in $\{1,\dots,Z\}$, and $\ell_j=L_j$,
$\tilde N_j=N_j$ for $2\le j\le R_n-1$.
Hence
\[
\HH(X^n\mid S,Y)\le1+2\log Z+\HH(L_2^{R_n-1}\mid R_n,N_2^{R_n-1}),
\]
where $L_2^{R_n-1}$ and $N_2^{R_n-1}$ are empty if $R_n\le2$.

\emph{Step 2: a deterministic number of runs.}
Let $\bar R=\max(R_n-1,k_n)$.
Since $L_2^{R_n-1}$ is a function of $(L_2^{\bar R},R_n)$,
\begin{align*}
&\HH(L_2^{R_n-1}\mid R_n,N_2^{R_n-1})\\
&\quad\le\HH(L_2^{\bar R}\mid R_n,N_2^{R_n-1})\\
&\quad\le\HH(L_2^{\bar R}\mid R_n,N_2^{\bar R})+\HH(N_{R_n}^{\bar R}\mid R_n,N_2^{R_n-1})\\
&\quad\le\HH(L_2^{k_n}\mid N_2^{k_n})+\HH(L_{k_n+1}^{\bar R}\mid R_n)+\HH(N_{R_n}^{\bar R}\mid R_n)\\
&\quad\le(k_n-1)\HH(L\mid N)+\E[(R_n-1-k_n)^+]\log Z\\
&\qquad+\E[(k_n-R_n+1)^+]\log(Z+1).
\end{align*}
The second inequality is $\HH(A\mid C)\le\HH(A\mid C,C')+\HH(C'\mid C)$ with
$C'=N_{R_n}^{\bar R}$; the third uses the chain rule and drops conditioning; the last uses the
independence of the pairs $(L_j,N_j)$, $j\ge2$, and counts the entries of $L_{k_n+1}^{\bar R}$ and
$N_{R_n}^{\bar R}$, whose numbers are functions of $R_n$ and whose values lie in $\{1,\dots,Z\}$ and
$\{0,\dots,Z\}$.

\emph{Step 3: the number of runs.}
If $R_n\ge k_n+2$, then $L_2+\dots+L_{k_n+1}<n$, while $\E[L_2+\dots+L_{k_n+1}]=k_n\mu\ge n+\mu n^{2/3}$.
By Hoeffding's inequality for $k_n$ independent variables with values in $[1,Z]$,
\[
\Pr[R_n\ge k_n+2]\le\exp\Bigl(-\frac{2\mu^2n^{4/3}}{k_n(Z-1)^2}\Bigr)\le e^{-cn^{1/3}}
\]
for some $c>0$ and all large $n$, because $k_n=O(n)$.
Since $R_n\le n$, $\E[(R_n-1-k_n)^+]\le ne^{-cn^{1/3}}$.
Next, $R_n$ is a stopping time for $(L_j)_{j\ge1}$, and $\theta_{R_n}\ge n$.
By Wald's identity~\cite{Durrett2019} applied to $L_2,L_3,\dots$,
$\E[\theta_{R_n}]=\E[L_1]+\mu\,\E[R_n-1]$, and therefore $\E[R_n]\ge1+(n-Z)/\mu$.
Hence
\begin{align*}
\E[(k_n-R_n+1)^+]&=k_n+1-\E[R_n]+\E[(R_n-k_n-1)^+]\\
&\le n^{2/3}+1+\frac{Z}{\mu}+ne^{-cn^{1/3}}.
\end{align*}

Combining the steps,
$\HH(X^n\mid S,Y)\le(k_n-1)\HH(L\mid N)+O(n^{2/3})$, and $k_n/n\to1/\mu$.
Finally, $\HH(L\mid N)=\HH(L,N)-\HH(N)=\HH(\Omega)-\HH(\omega)$.
\end{proof}

\subsection{A lower bound on $\HH(S\mid X^n,Y)$}

The last term of~\eqref{eq:vtr-split} is the uncertainty about which input runs produced the output
when both $X^n$ and $Y$ are known.
It arises whenever an output run collects the survivors of several input runs of the same symbol.
For $z\ge1$, $r\ge1$, and $s\ge0$, let
\begin{equation}
\label{eq:hsplit}
\hg(z,r,s)=\HH\Bigl(\Bigl\{\tbinom{z}{l}\tbinom{r}{s-l}\big/\tbinom{z+r}{s}\Bigr\}_{l}\Bigr)
\end{equation}
for $s\le z+r$ be the entropy of the hypergeometric law of the number of survivors among the first
$z$ of $z+r$ bits, given that $s$ of them survive.
If $W$ and $W'$ are independent with $W\sim\mathrm{Bin}(z,\db)$ and $W'\sim\mathrm{Bin}(r,\db)$, then
$\Pr[W=l\mid W+W'=s]=\binom{z}{l}\binom{r}{s-l}/\binom{z+r}{s}$, whose entropy is $\hg(z,r,s)$.
Since this law is supported on $\{0,\dots,\min(z,s)\}$, $\hg(z,r,0)=0$ and
$0\le\hg(z,r,s)\le\log(z+1)$.
Let
\begin{equation}
\label{eq:Kchain}
\Theta_{z,r}=\sum_{s=1}^{z+r}\Bin_{z+r}(s)\,\hg(z,r,s),
\end{equation}
which is the expected value of $\hg(z,r,s)$ when $s\sim\mathrm{Bin}(z+r,\db)$.

\begin{lemma}[Chain penalty]
\label{lem:chain}
For the stationary renewal input with finite-support law $P$,
\begin{equation}
\label{eq:phi-chain}
\liminf_{n\to\infty}\frac1n\HH(S\mid X^n,Y)\ge\Phi(P,d),
\end{equation}
where
\begin{equation}
\label{eq:phi-def}
\Phi(P,d):=\frac{1-D}{\mu}\sum_{z\ge1}\sum_{r\ge1}P_z\Psi_r \Theta_{z,r}.
\end{equation}
The series converges, and $\Phi(P,d)\le\log(Z+1)D/\mu$.
\end{lemma}

\begin{proof}
By Part~4 of Lemma~\ref{lem:runlevel}, $\HH(S\mid X^n,Y)=\HH(\tilde N\mid X^n,Y)$ with
$\tilde N=(\tilde N_1,\dots,\tilde N_{R_n})$.
Fix a realization $x$ of $X^n$; it fixes $R=R_n$, the lengths $\ell_1,\dots,\ell_R$, and the symbols
of the runs.
Probabilities and entropies conditional on $X^n=x$ are denoted by $\Pr_x$ and $\HH_x$.
Under $\Pr_x$ the variables $\tilde N_j$ are independent with $\tilde N_j\sim\mathrm{Bin}(\ell_j,\db)$
(Section~\ref{subsec:runs}), and $Y$ is a function of $\tilde N$.

\emph{Step 1: chains.}
For integers $e\ge1$ and $i\ge1$ with $e+2i+2\le R$, let $b_m=e+2m-1$ for $m=1,\dots,i+1$, and
let $C^{e,i}$ be the event
\[
\tilde N_e\ge1,\ \ \tilde N_{e+2}=\tilde N_{e+4}=\dots=\tilde N_{e+2i}=0,\ \ \tilde N_{e+2i+2}\ge1.
\]
Call the symbol of run $e$ $\mathtt a$ and the other symbol $\mathtt b$.
On $C^{e,i}$, if at least one of the $\mathtt b$-runs $b_1,\dots,b_{i+1}$ has a surviving bit, then the
first such bit starts an output run $o$ (the preceding output bit is the last survivor of run $e$,
which has symbol $\mathtt a$), and $o$ consists of the survivors of $b_1,\dots,b_{i+1}$, because the $\mathtt a$-runs
between them are deleted and run $e+2i+2$ has a surviving bit of symbol $\mathtt a$.
We call $e$ the \emph{anchor} and $(b_1,\dots,b_{i+1})$ the \emph{chain}.
For $1\le m\le i$ let
\[
r_m=\sum_{m'=m+1}^{i+1}\ell_{b_{m'}},\qquad s_m=\sum_{m'=m}^{i+1}\tilde N_{b_{m'}},
\]
the total length of $b_{m+1},\dots,b_{i+1}$ and the number of survivors of $b_m,\dots,b_{i+1}$.
For a fixed run $j$, the events $C^{e,i}$ with $j=b_m$ for some $1\le m\le i$ are disjoint:
if $C^{e,i}$ and $C^{e',i'}$ hold with $e<e'<j$ and $j-e$, $j-e'$ odd, then $e'$ is an $\mathtt a$-run
strictly between $e$ and $e+2i+2$ with $\tilde N_{e'}\ge1$, which contradicts $C^{e,i}$; and for a
fixed anchor $e$, $C^{e,i}$ forces $i$ to be the smallest integer with $\tilde N_{e+2i+2}\ge1$.

\emph{Step 2: a genie for one run.}
By the chain rule,
$\HH_x(\tilde N\mid Y)=\sum_{j=1}^{R}\HH_x(\tilde N_j\mid\tilde N^{j-1},Y)$.
Fix $j$ and let $\Upsilon_j=(e,i)$ if $C^{e,i}$ holds for the (unique) pair with $j=b_m$,
$1\le m\le i$, and $\Upsilon_j=\emptyset$ otherwise.
On $\{\Upsilon_j=(e,i)\}$, let $\mathcal B=\{b_m,\dots,b_{i+1}\}$ with $j=b_m$, let $V$ be the vector
$(\tilde N_{j'})_{j'\notin\mathcal B}$, and let $\Gamma_j=(\Upsilon_j,V,s_m)$; on
$\{\Upsilon_j=\emptyset\}$ let $\Gamma_j=\Upsilon_j$.
Since conditioning does not increase entropy,
\begin{align*}
\HH_x(\tilde N_j\mid\tilde N^{j-1},Y)
&\ge\HH_x(\tilde N_j\mid\tilde N^{j-1},Y,\Gamma_j)\\
&\ge\sum_{(e,i)}\Pr_x[\Upsilon_j=(e,i)]\\
&\quad\times\HH_x\bigl(\tilde N_j\mid\tilde N^{j-1},Y,V,s_m,
\\
&\qquad\qquad\Upsilon_j=(e,i)\bigr),
\end{align*}
where the terms with $\Upsilon_j=\emptyset$ were dropped because they are nonnegative.

\emph{Step 3: the conditional law.}
Fix $(e,i)$ and $m$ with $j=b_m$, and the corresponding $\mathcal B$ and $V$.
The event $C^{e,i}$ involves only runs outside $\mathcal B$; hence it is determined by $V$.
The prefix $\tilde N^{j-1}$ is a sub-vector of $V$, because all indices in $\mathcal B$ are at least
$j$.
The output $Y$ depends on $(\tilde N_b)_{b\in\mathcal B}$ only through $s_m$, because on $C^{e,i}$ the
survivors of the runs in $\mathcal B$ are consecutive bits of the same output run $o$ (the $\mathtt a$-runs
between them being deleted), and only their number matters; hence $Y$ is a function of $(V,s_m)$ on
$C^{e,i}$.
Therefore, on $\{\Upsilon_j=(e,i)\}=C^{e,i}$, conditioning on $(\tilde N^{j-1},Y,V,s_m)$ is the same as
conditioning on $(V,s_m)$.
Under $\Pr_x$, $(\tilde N_b)_{b\in\mathcal B}$ is independent of $V$, and its components are
independent binomial variables with lengths $\ell_b$.
Hence, conditionally on $V$ and on $s_m=s$, the variable $\tilde N_j$ is the number of survivors among
the $z=\ell_j$ bits of run $j$, given that $s$ of the $z+r_m$ bits of the runs in $\mathcal B$
survive; its conditional law is the hypergeometric law in~\eqref{eq:hsplit}, and
\[
\HH_x(\tilde N_j\mid V,s_m,C^{e,i})=\E_x\bigl[\hg(\ell_j,r_m,s_m)\bigm|C^{e,i}\bigr].
\]
Summing over $j$ and then over the pairs $(e,i)$ and $m\le i$ gives
\begin{equation}
\label{eq:chain-sum}
\HH_x(\tilde N\mid Y)\ge\E_x\Bigl[\sum_{(e,i)}\ones_{C^{e,i}}\sum_{m=1}^{i}\hg(\ell_{b_m},r_m,s_m)\Bigr].
\end{equation}
(The event that no run of the chain survives is not excluded from $C^{e,i}$; on it $s_m=0$ and
$\hg(\ell_{b_m},r_m,0)=0$.)

\emph{Step 4: complete runs.}
Average~\eqref{eq:chain-sum} over $X^n$ and keep only the pairs with $e\ge2$ and
$e+2i+2\le R_n-1$; the omitted terms are nonnegative.
All runs involved are then complete, so that $\ell_j=L_j$ and $\tilde N_j=N_j$.
For $e\ge2$ and $i\ge1$ let
\begin{multline*}
F_{e,i}=\ones\{N_e\ge1\}\prod_{k=1}^{i}\ones\{N_{e+2k}=0\}
\\
\quad\times\ones\{N_{e+2i+2}\ge1\}
\sum_{m=1}^{i}\hg(L_{b_m},r_m,s_m),
\end{multline*}
with $r_m$ and $s_m$ defined as in Step~1 from $(L_j,N_j)$.
For fixed $e$, at most one $i$ gives $F_{e,i}\ne0$, and $F_{e,i}\le i\log(Z+1)$.
Let $m_n=\lfloor n/\mu-n^{2/3}\rfloor$.
On $\{R_n\ge m_n+1\}$, every pair with $e+2i+2\le m_n$ satisfies $e+2i+2\le R_n-1$.
Hence
\begin{multline*}
\HH(S\mid X^n,Y)\ge\E\Bigl[\sum_{\substack{e\ge2,\,i\ge1\\e+2i+2\le m_n}}F_{e,i}\Bigr]
\\
\quad-m_n^2\log(Z+1)\Pr[R_n\le m_n].
\end{multline*}
If $Z=1$, then $R_n=n>m_n$ and the last term is zero.
If $Z\ge2$ and $R_n\le m_n$, then $\theta_{m_n}\ge n$ and $L_2+\dots+L_{m_n}\ge n-Z$, while
$\E[L_2+\dots+L_{m_n}]=(m_n-1)\mu\le n-\mu n^{2/3}-\mu$.
The deviation is at least $\mu n^{2/3}+\mu-Z$, and Hoeffding's inequality gives
\[
\Pr[R_n\le m_n]\le\exp\Bigl(-\frac{2(\mu n^{2/3}+\mu-Z)^2}{(m_n-1)(Z-1)^2}\Bigr)\le e^{-cn^{1/3}}
\]
for some $c>0$ and all large $n$; the last term therefore vanishes as $n\to\infty$.

\emph{Step 5: expectation of one term.}
The pairs $(L_j,N_j)$, $j\ge2$, are i.i.d.; hence $\phi_i=\E[F_{e,i}]$ does not depend on $e\ge2$.
The factors of $F_{e,i}$ involve disjoint sets of runs: run $e$, the $\mathtt a$-runs $e+2,\dots,e+2i$, run
$e+2i+2$, and the $\mathtt b$-runs $b_1,\dots,b_{i+1}$.
Hence
\[
\phi_i=(1-D)\,D^{\,i}\,(1-D)\sum_{m=1}^{i}\E\bigl[\hg(L_{b_m},r_m,s_m)\bigr].
\]
For $1\le m\le i$, let $h=i+1-m\ge1$ be the number of $\mathtt b$-runs after $b_m$.
The length $L_{b_m}$ has law $P$, the total length $r_m$ of the next $h$ runs has law
$Q_{\cdot,h}$, and, given these lengths, $s_m\sim\mathrm{Bin}(L_{b_m}+r_m,\db)$.
Hence $\E[\hg(L_{b_m},r_m,s_m)]=a_h$ with $a_h=\sum_{z,r}P_zQ_{r,h}\Theta_{z,r}$, and
\[
\phi_i=(1-D)^2D^{\,i}\sum_{h=1}^{i}a_h\le(1-D)^2D^{\,i}\,i\log(Z+1).
\]

\emph{Step 6: counting.}
For each $i\ge1$, the number of $e\ge2$ with $e+2i+2\le m_n$ is $(m_n-2i-3)^+$.
Hence the expectation in Step~4 equals $\sum_{i\ge1}(m_n-2i-3)^+\phi_i$.
Since $(m_n-2i-3)^+/n\le1/\mu$, $(m_n-2i-3)^+/n\to1/\mu$ for each $i$, and $\sum_i\phi_i<\infty$, dominated
convergence gives
\[
\liminf_{n\to\infty}\frac1n\HH(S\mid X^n,Y)\ge\frac1\mu\sum_{i\ge1}\phi_i.
\]
Exchanging the order of summation,
\begin{align*}
\sum_{i\ge1}\phi_i&=(1-D)^2\sum_{h\ge1}a_h\sum_{i\ge h}D^{\,i}=(1-D)\sum_{h\ge1}D^{\,h}a_h\\
&=(1-D)\sum_{z\ge1}\sum_{r\ge1}P_z\Bigl(\sum_{h\ge1}D^{\,h}Q_{r,h}\Bigr)\Theta_{z,r}.
\end{align*}
For $r\ge1$, $Q_{r,0}=0$, and the inner sum equals $\Psi_r$.
This proves~\eqref{eq:phi-chain}--\eqref{eq:phi-def}.
Finally, $\Theta_{z,r}\le\log(Z+1)$ for $z\le Z$ and $\sum_{r\ge1}\Psi_r=D/(1-D)$ by
Lemma~\ref{lem:psi}, which gives the bound on $\Phi$.
\end{proof}

\begin{proposition}[Comparison with the VTR correction]
\label{prop:vtr-compare}
Let $\Phi_{\mathrm{VTR}}(P,d)=\rho D\sum_{z,r\ge1}P_zP_r\Theta_{z,r}$ with $\rho$ from~\eqref{eq:rho}.
Then, for every finite-support law $P$,
\begin{multline}
\label{eq:phi-compare}
\Phi(P,d)=(1+D)\Phi_{\mathrm{VTR}}(P,d)
\\
\quad+\frac{1-D}{\mu}\sum_{z,r\ge1}P_z
\Bigl(\sum_{j\ge2}D^{\,j}Q_{r,j}\Bigr)\Theta_{z,r},
\end{multline}
and in particular $\Phi(P,d)\ge(1+D)\Phi_{\mathrm{VTR}}(P,d)$.
For the geometric law $P_z=(1-\gamma)\gamma^{z-1}$, $z\ge1$,
\begin{multline}
\label{eq:phi-vtr-geo}
\Phi_{\mathrm{VTR}}=\frac{\rho(1-\gamma)^3d}{\gamma^2(1-\gamma d)}
\sum_{z,r\ge1}(\gamma d)^{z+r}
\\
\quad\times\sum_{s=1}^{z+r}\binom{z+r}{s}
\Bigl(\frac{\db}{d}\Bigr)^{s}\hg(z,r,s),
\end{multline}
and $\rho=\db(1-q_\gamma)$ with $q_\gamma=(\gamma+d-2\gamma d)/(1+d-2\gamma d)$.
\end{proposition}

\begin{proof}
For $r\ge1$, $\Psi_r=DQ_{r,1}+\sum_{j\ge2}D^{\,j}Q_{r,j}$ and $Q_{r,1}=P_r$.
Substituting into~\eqref{eq:phi-def} and using
$\frac{1-D}{\mu}D=(1+D)\rho D$ gives~\eqref{eq:phi-compare}; the last sum is nonnegative.
For the geometric law, $P_zP_r=(1-\gamma)^2\gamma^{z+r-2}$ and
$\Bin_{z+r}(s)=d^{\,z+r}\binom{z+r}{s}(\db/d)^s$, which gives the double sum
in~\eqref{eq:phi-vtr-geo}, and
$D=\sum_z(1-\gamma)\gamma^{z-1}d^{\,z}=(1-\gamma)d/(1-\gamma d)$.
Moreover $\mu=1/(1-\gamma)$, $1-D=\db/(1-\gamma d)$, and $1+D=(1+d-2\gamma d)/(1-\gamma d)$; hence
$\rho=(1-\gamma)\db/(1+d-2\gamma d)=\db(1-q_\gamma)$.
\end{proof}

The expression~\eqref{eq:phi-vtr-geo} coincides with the correction term
of~\cite[Thm.~3]{VenkataramananTatikondaRamchandran2013}, in which $\db(1-q_\gamma)$ is the rate of
output runs per input bit.
The formula~\eqref{eq:phi-vtr-geo} is an evaluation of the defining expression of
$\Phi_{\mathrm{VTR}}$ for a law of infinite support; it is used only to identify $\Phi_{\mathrm{VTR}}$ with
the correction term of~\cite{VenkataramananTatikondaRamchandran2013}.
Proposition~\ref{prop:vtr-compare} states that Lemma~\ref{lem:chain} improves this correction by the
factor $1+D$ and by the contribution of chains with more than two $\mathtt b$-runs, for every finite-support
law, including truncated geometric laws.
Section~\ref{sec:results} quantifies the resulting gain for truncated geometric laws.

\subsection{The bound}

\begin{theorem}[Free-law VTR bound]
\label{thm:vtr}
For every finite-support run-length law $P$, every $0<d<1$, and every integer $\kappa\ge1$,
\begin{multline}
\label{eq:vtr-lb}
\Cdel\ge\RVTR(P,d;\kappa):=\frac{\HH(P)}{\mu}-\db\,\HH_\pi(S_t\mid Y_{t-\kappa},\dots,Y_t)\\
-\frac{\HH(L\mid N)}{\mu}+\Phi(P,d).
\end{multline}
\end{theorem}

\begin{proof}
Let $X^n$ be the window of the stationary renewal input with law $P$.
Divide~\eqref{eq:vtr-split} by $n$.
Since $\liminf(a_n+b_n)\ge\liminf a_n+\liminf b_n$ and $\liminf(-a_n)=-\limsup a_n$,
\begin{align*}
\liminf_{n\to\infty}\frac1n\I(X^n;Y)
&\ge\lim_{n}\frac1n\HH(X^n)-\limsup_n\frac1n\HH(S\mid Y)\\
&\quad-\limsup_n\frac1n\HH(X^n\mid S,Y)\\
&\quad+\liminf_n\frac1n\HH(S\mid X^n,Y).
\end{align*}
The four terms are bounded by~\eqref{eq:entropy-rate}, Lemma~\ref{lem:hsy}, Lemma~\ref{lem:hxsy},
and Lemma~\ref{lem:chain}, which gives $\liminf_n\frac1n\I(X^n;Y)\ge\RVTR(P,d;\kappa)$, and
\eqref{eq:liminf-C} completes the proof.
\end{proof}

All terms of~\eqref{eq:vtr-lb} except $\Phi$ are finite sums for a finite-support law:
$\HH(P)$ and $\HH(L\mid N)$ are sums over at most $Z(Z+1)$ terms, and~\eqref{eq:hsy-filter} is a sum
over $2^\kappa$ sequences, each requiring $\kappa$ products of a row vector with a $Z\times Z$ matrix.
The series $\Phi$ has nonnegative terms, and any finite partial sum is a lower bound; replacing $\Phi$
by a partial sum in~\eqref{eq:vtr-lb} therefore gives a valid lower bound on $\Cdel$.
The geometric laws used below as a Markov control are truncated to $\{1,\dots,Z\}$ and renormalized,
which gives admissible finite-support laws; no limit over the truncation is taken.

%% file: sec_numerics_TIT.tex
\section{Optimization and Verified Evaluation}
\label{sec:numerics}

Theorems~\ref{thm:dm} and~\ref{thm:vtr} hold for every finite-support law.
The numerical task has two independent parts: a search that proposes a law $P$ with a large value of
the functional, and a verified evaluation of the functional at that law
(Fig.~\ref{fig:pipeline}).
Only the second part affects validity: an imperfect search can only produce a weaker bound, never an
invalid one.

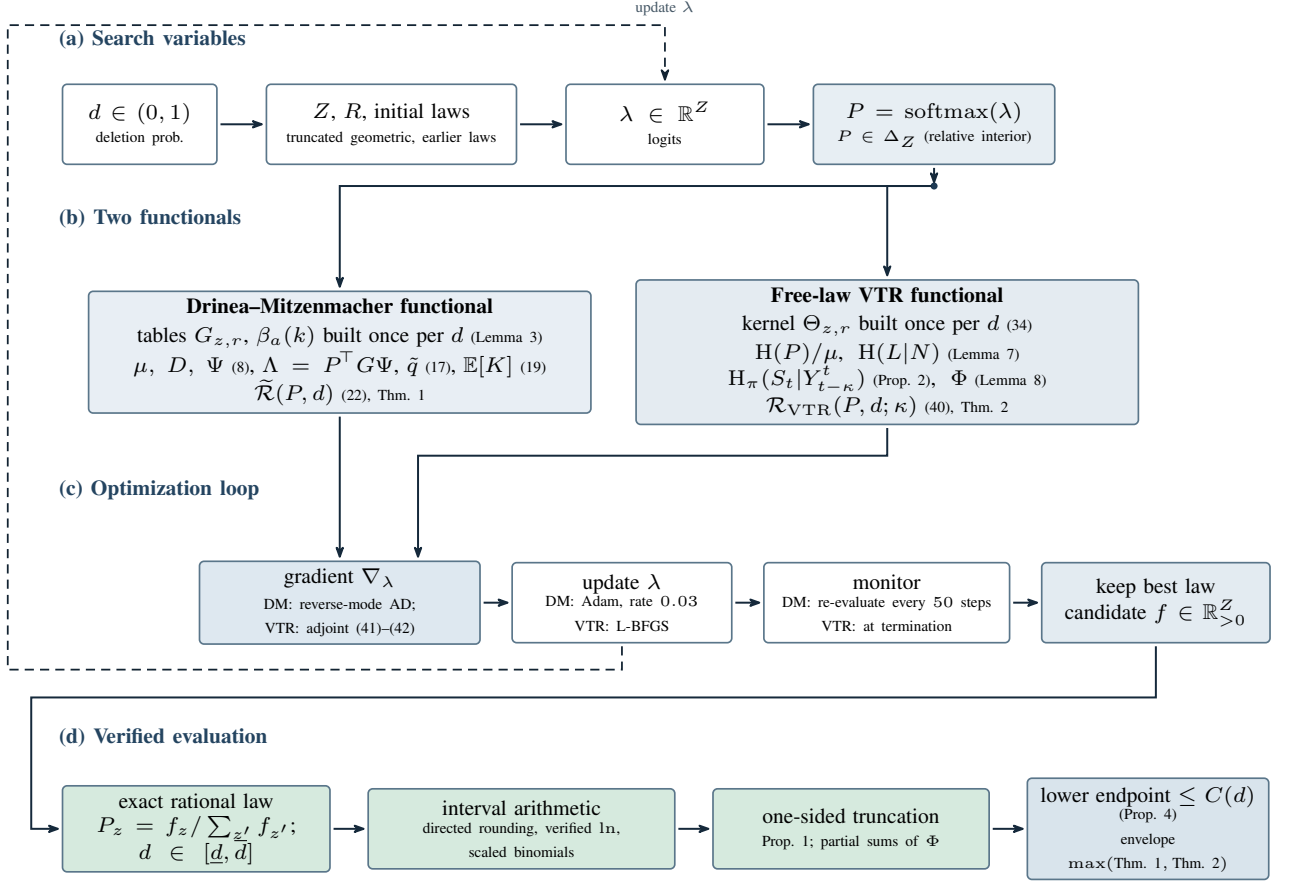
\begin{figure*}[!t]
\centering
\resizebox{\textwidth}{!}{\input{figures/pipeline_TIT}}
\caption{Search and verification.
The search proposes a law in double-precision arithmetic; the reported bound is the lower endpoint
of an interval enclosure of the functional at the exact rational law defined by the candidate, and
its validity rests only on the theorems and on Proposition~\ref{prop:verified}.
(a)~Unconstrained logits $\lambda$ are mapped to a law on $\{1,\dots,Z\}$.
(b)~For each $d$ the tables that depend on $d$ only are built once and held fixed; each evaluation of
either functional is a short sequence of recursions and bilinear forms in $(P,\Psi)$.
(c)~The DM search uses Adam with reverse-mode differentiation and monitors the law every $50$ steps;
the VTR search uses L-BFGS with the adjoint gradient.
No codeword or channel output is sampled at any stage.
(d)~The best law is re-evaluated at the exact rational law $f/\sum f$ in interval arithmetic; the
reported bound is the larger of the two verified lower endpoints.}
\label{fig:pipeline}
\end{figure*}

\subsection{The search problem}
\label{subsec:search-problem}

For a fixed $d$ and support $\{1,\dots,Z\}$ we seek laws with large $\RDM(P,d)$ or $\RVTR(P,d;1)$,
that is, approximate maximizers over the simplex
$\Delta_Z=\{P\in\R^Z_{\ge0}:\sum_zP_z=1\}$.
Both objectives are nonconcave, and no claim of global or local optimality is made.
We parameterize the relative interior of $\Delta_Z$ by $P=\mathrm{softmax}(\lambda)$,
$P_z=e^{\lambda_z}/\sum_{z'}e^{\lambda_{z'}}$, $\lambda\in\R^Z$, and ascend along gradients of finite
truncations of the objectives.
If $\partial f/\partial P$ denotes the gradient with respect to $P$ (with the coordinates of $P$
treated as independent), then
$\partial f/\partial\lambda_z=P_z\bigl(\partial f/\partial P_z-\sum_{z'}P_{z'}\partial f/\partial P_{z'}\bigr)$.
The quantities $\mu$ and $D$ are linear in $P$, with $\partial\mu/\partial P_z=z$ and
$\partial D/\partial P_z=d^{\,z}$.

Both objectives contain linear functionals $F=\sum_{r=1}^{R}c_r\Psi_r$ of the truncated renewal weights,
where $c_r$ may depend on $P$ explicitly.
The derivative of $F$ through $\Psi$ is obtained from the adjoint recursion
\begin{equation}
\label{eq:adjoint}
\xi_r=c_r+D\sum_{\ell=1}^{Z}P_\ell\,\xi_{r+\ell},\qquad r=R,R-1,\dots,1,
\end{equation}
with $\xi_{r'}=0$ for $r'>R$; then
\begin{equation}
\label{eq:adjoint-grad}
\frac{\partial F}{\partial P_\ell}\Big|_{\Psi}=D\sum_{r=\ell}^{R}\xi_r\Psi_{r-\ell},
\qquad
\frac{\partial F}{\partial D}\Big|_{\Psi}=\frac1D\sum_{r=1}^{R}\xi_r\Psi_r .
\end{equation}
Indeed, differentiating~\eqref{eq:psi-rec} gives the triangular linear system
$\mathrm d\Psi_r-D\sum_\ell P_\ell\,\mathrm d\Psi_{r-\ell}=b_r$ with
$b_r=\mathrm dD\sum_\ell P_\ell\Psi_{r-\ell}+D\sum_\ell\mathrm dP_\ell\,\Psi_{r-\ell}$.
Writing it as $(I-M)\,\mathrm d\Psi=b$, we get $\mathrm dF=c^{\top}(I-M)^{-1}b=\xi^{\top}b$ with
$\xi=(I-M^{\top})^{-1}c$, which is~\eqref{eq:adjoint}; collecting the coefficients of
$\mathrm dP_\ell$ and $\mathrm dD$ in $\xi^{\top}b$, and using
$\sum_\ell P_\ell\Psi_{r-\ell}=\Psi_r/D$, gives~\eqref{eq:adjoint-grad}.
The cost is $O(RZ)$ operations, the same as the recursion~\eqref{eq:psi-rec}.

\emph{DM functional.}
By Lemma~\ref{lem:dm-reduce}, the tables $G_{z,r}$ and $\Bin_a(k)$ are computed once per $d$, and
each evaluation requires $\Psi$, the bilinear form $P^{\top}G\Psi$, and the
mixture~\eqref{eq:q-bilinear}.
In the search, $\Psi$ is computed as the truncated sum $\sum_{i\le I}D^iP^{*i}$ of convolution powers
with $D^{I}\le10^{-14}$, and the gradient is obtained by reverse-mode automatic differentiation of this
graph.

\emph{VTR functional.}
The objective is $\RVTR(P,d;1)$ with $\Phi$ truncated to $r\le R$.
The term $\HH(P)/\mu$ has the explicit gradient
$-(\log P_z+\log e)/\mu-z\HH(P)/\mu^2$.
The pairwise term~\eqref{eq:hsy-pair} depends on $P$ only through $(\mu,D)$, and its partial
derivatives with respect to $\mu$ and $D$ are computed by central differences.
The derivative of $\HH(L\mid N)=\HH(\Omega)-\HH(\omega)$ with respect to $P_r$ is
$-\sum_m\Bin_r(m)\log(\Omega_{r,m}/\omega_m)$.
The chain penalty is $\Phi=c\,F$ with $c=(1-D)/\mu$ and $F=\sum_{z,r}P_z\Psi_r\Theta_{z,r}$; its gradient
combines $\partial c/\partial P_z=-d^{\,z}/\mu-(1-D)z/\mu^2$, the explicit derivative
$\sum_r\Psi_r\Theta_{z,r}$ of $F$, and~\eqref{eq:adjoint-grad} with $c_r=\sum_zP_z\Theta_{z,r}$.
The gradients were checked against finite differences of the objective.

\subsection{Search protocols}
\label{subsec:search}

\emph{DM functional (Algorithm~\ref{alg:dm}).}
For each $d$, let $\hat p(d)$ be the stay probability of~\cite[Table~I]{DrineaMitzenmacher2007},
linearly interpolated off the tabulated grid with $\hat p(0)=1/2$, and $\mu_0=1/(1-\hat p)$.
The main runs used supports $Z\approx20\mu_0$, capped at $1000$ (the cap is active at $d=0.85$ and
$d=0.9$; $Z=2400$ at $d=0.95$), and at about half of the values of $d$ a support about $1.4$ times
larger gave a larger value and was retained; the truncation $R$ of the renewal weights in $\Lambda$ is
between $1.18Z$ and $1.6Z$, which leaves an omitted renewal mass below $10^{-13}$.
The values of $Z$ and $R$ for every $d$ are listed with the supplementary material.
Three restarts start from truncated geometric laws with stay probabilities $\hat p$ and
$\hat p\pm0.08(1-\hat p)$; each runs $600$ Adam steps~\cite{KingmaBa2015} (learning rate $0.03$), and
every $50$ steps the current law is evaluated in double precision; the best law, including the
geometric law with stay probability $\hat p$, is retained.
Longer polishing runs changed the double-precision values by at most $1.1\times10^{-5}$~bits; the
laws that are verified and reported are those of the main runs, which are the laws reported in the
earlier version of this work~\cite{khodaiemehr2026improvedlowerboundscapacity}.
Exploratory runs with mixtures of geometric and Gaussian components, sparsity penalties, and
exponentiated-gradient updates, reported in~\cite{khodaiemehr2026improvedlowerboundscapacity}, did not improve on these laws by more than
$4\times10^{-6}$~bits (Section~\ref{sec:discussion}).

\begin{algorithm*}[!t]
\caption{Search for the Drinea--Mitzenmacher functional (Theorem~\ref{thm:dm})}
\label{alg:dm}
\small
\begin{algorithmic}[1]
\Require deletion probability $d$; stay probability $\hat p(d)$; support $Z$; truncation $R$;
  $T=600$ steps; learning rate $\eta=0.03$; monitoring period $50$
\Ensure candidate vector $f=P^\star\in\R^Z_{>0}$ (verified by Algorithm~\ref{alg:verify})
\Statex \emph{Tables that depend on $d$ only (Lemma~\ref{lem:dm-reduce}), built once:}
\State $G_{z,r}\gets\sum_{k=1}^{z+r}w(z,r,k)\log[\binom{z+r}{k}-\binom{r}{k}]$ for $1\le z\le Z$, $0\le r\le R$
  \Comment{\eqref{eq:w}--\eqref{eq:G}}
\State $\Bin_a(k)\gets\binom ak\db^{\,k}d^{\,a-k}$ for $0\le a\le Z+R$, $0\le k\le a$
\Statex \emph{Objective $f(P)$ used by the search, for a law $P$ on $\{1,\dots,Z\}$:}
\State $\mu\gets\sum_zzP_z$;\ $D\gets\sum_zP_zd^{\,z}$;\ $\rho\gets(1-D)/(\mu(1+D))$
  \Comment{\eqref{eq:mu-D}, \eqref{eq:rho}}
\State $\Psi\gets\sum_{i=0}^{I}D^iP^{*i}$ truncated to $r\le Z+R$, with $I$ such that $D^I\le10^{-14}$
  \Comment{$P^{*i}$: $i$-fold convolution; \eqref{eq:Q-Psi}}
\State $\Lambda\gets\sum_{z\le Z}\sum_{r\le R}P_z\Psi_rG_{z,r}$ \Comment{\eqref{eq:lambda-bilinear}}
\State $\tilde q_k\gets\max\bigl(0,\sum_{a\le Z+R}\bigl[(P*\Psi)_a-D\Psi_a\bigr]\Bin_a(k)\bigr)$
  \Comment{first form of \eqref{eq:q-bilinear} via \eqref{eq:w-split}}
\State $\varepsilon\gets\max(0,1-\sum_k\tilde q_k)$;\
  $f(P)\gets-h(d)+\rho\bigl(-\sum_k\tilde q_k\log\tilde q_k+(1-\varepsilon)\log(1-\varepsilon)+\Lambda\bigr)$
  \Comment{\eqref{eq:dm-trunc}}
\Statex \emph{Search:}
\State $P^\star\gets$ geometric law with stay probability $\hat p$, truncated to $\{1,\dots,Z\}$ and
  renormalized; $f^\star\gets f(P^\star)$
\For{$p\in\{\hat p,\ \hat p+0.08(1-\hat p),\ \hat p-0.08(1-\hat p)\}$} \Comment{three restarts}
  \State $\lambda\gets\log$ of the truncated geometric law with stay probability $p$;
    reset the Adam moments
  \For{$t=1,\dots,T$}
    \State $P\gets\mathrm{softmax}(\lambda)$; evaluate $f(P)$ as in lines 3--7
    \State compute $\nabla_\lambda f$ by reverse-mode differentiation of lines 3--7
    \State $\lambda\gets\lambda+\mathrm{Adam}_\eta(\nabla_\lambda f)$ \Comment{ascent step}
    \If{$t\equiv0\pmod{50}$ or $t=T$}
      \State $P\gets\mathrm{softmax}(\lambda)$; evaluate $f(P)$ afresh in double precision
      \If{$f(P)>f^\star$} $P^\star\gets P$;\ $f^\star\gets f(P)$ \EndIf
    \EndIf
  \EndFor
\EndFor
\State \Return $P^\star$
\end{algorithmic}
\end{algorithm*}

\emph{VTR functional (Algorithm~\ref{alg:vtr}).}
Let $\hat\gamma$ be the stay probability of~\cite[Table~I]{VenkataramananTatikondaRamchandran2013}
at the same $d$ (interpolated for $d<0.05$ and $0.05<d<0.1$; for $d\ge0.75$ the value $\hat p(d)$ above).
The support is $Z=\max(40,\lceil15/(1-\hat\gamma)\rceil)$, capped at $160$, and $R=2Z$.
The initial laws are truncated geometric laws with stay probabilities
$\min(0.99,\max(0.3,\hat\gamma+\delta))$, $\delta\in\{0,\pm0.02,\pm0.05\}$, the law found for the DM
functional at the same $d$, and laws found in earlier runs of the search whose objectives used other
correction terms: the correction term of the earlier version~\cite{khodaiemehr2026improvedlowerboundscapacity} and the penalty $(1-D)\Phi(P,d)$.
An initial law with a larger support $Z_{\mathrm{seed}}>Z$ (the DM laws have supports up to $2400$) is
restricted to its first $Z$ coordinates; a law with a smaller support is padded with zeros; in both
cases every mass is then floored at $10^{-14}$ and the vector is renormalized.
This choice affects only the search, not the validity of the verified value.
L-BFGS~\cite{LiuNocedal1989} runs for at most $1500$ iterations from each initial law (memory $30$), and
the law with the largest double-precision objective is verified.
The winning initial law was geometric at $13$ of the $26$ values of $d$, an earlier run at $12$, and
the DM law at one.

\begin{algorithm*}[!t]
\caption{Search for the free-law VTR functional (Theorem~\ref{thm:vtr})}
\label{alg:vtr}
\small
\begin{algorithmic}[1]
\Require deletion probability $d$; stay probability $\hat\gamma(d)$; support
  $Z=\min(160,\max(40,\lceil15/(1-\hat\gamma)\rceil))$; $R=2Z$; iteration limit $1500$
\Ensure candidate vector $f=P^\star\in\R^Z_{>0}$ (verified by Algorithm~\ref{alg:verify}, $\kappa\in\{1,3\}$)
\Statex \emph{Tables that depend on $d$ only, built once:}
\State $\Theta_{z,r}\gets$ midpoint of the interval enclosure of~\eqref{eq:Kchain}, $1\le z\le Z$,
  $1\le r\le R$ \Comment{sum over $s$ within $12$ std.\ dev.}
\State $\Bin_r(m)$ for $1\le r\le Z$, $0\le m\le r$
\Statex \emph{Objective $f(P)=\RVTR(P,d;1)$ with $\Phi$ truncated to $r\le R$, and its gradient $g=\partial f/\partial P$:}
\State $\mu\gets\sum_zzP_z$;\ $D\gets\sum_zP_zd^{\,z}$;\ $\alpha\gets(1-D)/(\mu\db)$;\
  $p_f\gets\alpha/(1+D)$;\ $p_0\gets(1-\alpha)/(1-p_f)$ \Comment{Lemma~\ref{lem:residual}}
\State $H_S\gets(1-p_f)[h(p_0)+(1-p_0)g]+p_fg$ with $g=h(D^2)/(1-D^2)$ \Comment{\eqref{eq:hsy-pair}}
\State $\Omega_{r,m}\gets P_r\Bin_r(m)$;\ $\omega_m\gets\sum_r\Omega_{r,m}$;\
  $H_{LN}\gets\HH(\Omega)-\HH(\omega)$ \Comment{Lemma~\ref{lem:hxsy}}
\State $\Psi_0\gets1$;\ $\Psi_r\gets D\sum_{\ell\le\min(r,Z)}P_\ell\Psi_{r-\ell}$ for $1\le r\le R$
  \Comment{\eqref{eq:psi-rec}}
\State $F\gets\sum_{z}\sum_{r=1}^{R}P_z\Psi_r\Theta_{z,r}$;\ $\Phi\gets(1-D)F/\mu$ \Comment{\eqref{eq:phi-def}}
\State $f(P)\gets\HH(P)/\mu-\db H_S-H_{LN}/\mu+\Phi$ \Comment{\eqref{eq:vtr-lb}}
\State $\xi_r\gets\sum_zP_z\Theta_{z,r}+D\sum_\ell P_\ell\xi_{r+\ell}$ for $r=R,\dots,1$ ($\xi_{r'}=0$, $r'>R$)
  \Comment{adjoint \eqref{eq:adjoint}}
\State $g_z\gets\partial_{P_z}[\HH(P)/\mu]-\db\,\partial_{P_z}H_S-\partial_{P_z}[H_{LN}/\mu]
  +\partial_{P_z}\Phi$ \Comment{Section~\ref{subsec:search-problem}, \eqref{eq:adjoint-grad}}
\Statex \emph{Search:}
\State initial laws: truncated geometric laws with stay probabilities
  $\min(0.99,\max(0.3,\hat\gamma+\delta))$, $\delta\in\{0,\pm0.02,\pm0.05\}$; the DM law of
  Algorithm~\ref{alg:dm} at the same $d$ and laws of earlier runs (each, with support $Z_{\mathrm{seed}}$, restricted to its first $\min(Z,Z_{\mathrm{seed}})$
  coordinates, zero-padded to $\{1,\dots,Z\}$ if $Z_{\mathrm{seed}}<Z$, floored at $10^{-14}$, and renormalized)
\For{each initial law $P^{(0)}$}
  \State $\lambda\gets\log P^{(0)}$
  \State minimize $\lambda\mapsto-f(\mathrm{softmax}(\lambda))$ by L-BFGS (memory $30$, at most $1500$
    iterations) with gradient $-P\odot(g-\langle P,g\rangle\ones)$ from lines 3--10
  \State record $P=\mathrm{softmax}(\lambda)$ at termination and $f(P)$
\EndFor
\State \Return the recorded law with the largest $f(P)$ as $P^\star$
\end{algorithmic}
\end{algorithm*}

\subsection{Verified evaluation}
\label{subsec:verified}

Each reported value is obtained as follows.
The search returns a vector $f\in\R^Z_{>0}$ of double-precision numbers.
The evaluated law is $P_z=f_z/\sum_{z'}f_{z'}$, an exact rational probability vector; it is an
admissible law in Theorems~\ref{thm:dm} and~\ref{thm:vtr} in its own right, and no approximation of
another law is involved.
The decimal value of $d$ is enclosed between the two neighbouring double-precision numbers.

A real number $x$ is represented by an interval $[\underline x,\overline x]$ with double-precision
endpoints and $\underline x\le x\le\overline x$.
Most real numbers, and most results of arithmetic operations on doubles, are not doubles; the
IEEE~754 standard~\cite{IEEE754} requires that each basic operation ($+$, $-$, $\times$, $\div$) return
the exact result rounded to a double according to a selectable rounding direction.
In the round-toward-$-\infty$ (rounding-down) direction, the returned value $\mathrm{RD}(a\circ b)$ is
the largest double that is at most the exact value $a\circ b$; in the round-toward-$+\infty$ direction,
$\mathrm{RU}(a\circ b)$ is the smallest double that is at least $a\circ b$.
Sums, differences, products, and quotients of intervals are computed with the lower endpoint
rounded down and the upper endpoint rounded up~\cite{MooreKearfottCloud2009}, for example
$[\underline a,\overline a]+[\underline b,\overline b]
=[\mathrm{RD}(\underline a+\underline b),\mathrm{RU}(\overline a+\overline b)]$.
The program sets the processor to the round-toward-$-\infty$ direction once and obtains upper
endpoints from the identity $\mathrm{RU}(a+b)=-\mathrm{RD}((-a)-b)$ (and its analogues for the other
operations), so no switching of rounding directions is needed.
Each such operation returns an interval that contains the exact result of the operation for all
arguments in the input intervals.

The only transcendental function is the logarithm.
For a positive double $x$, write $x=m\,2^{e}$ exactly with $m\in[2^{-1/2},2^{1/2})$ and an integer
$e$ (the comparison with $2^{-1/2}$ uses a double-precision constant, which only moves $m$ within
$[0.7071,1.4143)$), and let $t=(m-1)/(m+1)$, so that $|t|<0.172$ and
\[
\ln x=e\ln2+2\operatorname{artanh}t,\qquad
\operatorname{artanh}t=\sum_{j\ge0}\frac{t^{2j+1}}{2j+1}.
\]
For $0\le t<1$, the partial sum over $j\le13$ is a lower bound, and adding
\[
\sum_{j\ge14}\frac{t^{2j+1}}{2j+1}\le\frac{t^{29}}{29(1-t^2)}
\]
gives an upper bound; for $t<0$ we use $\operatorname{artanh}t=-\operatorname{artanh}|t|$.
All operations are rounded outward, $t$ is itself enclosed, and $\ln2$ is enclosed between two
consecutive doubles.
For an interval argument, monotonicity of $\ln$ gives the enclosure from the two endpoints.
For $\phi(p)=-p\log p$, the argument interval is first intersected with $[0,1]$; the lower endpoint
of the enclosure is the smaller of the lower enclosures of $\phi$ at the two endpoints (since $\phi$ is
concave), and the upper endpoint is the larger of the upper enclosures at the two endpoints if the
interval lies in $[0,1/e]$ or in $[1/e,1]$, where $\phi$ is monotone, and an upper bound on
$\max\phi=(\log e)/e$ otherwise.
The value $\phi(0)=0$ is returned directly, without evaluating $\log0$; the tests for
$[0,1/e]$ and $[1/e,1]$ use the doubles $0.3678794411714<1/e$ and $0.3679>1/e$; and the upper bound
on $\max\phi\approx0.530738$ is the constant $0.5308$.
Divisions abort if the divisor interval is not positive, and a logarithm whose argument interval is
not positive returns $-\infty$ as its lower endpoint, which keeps the enclosure valid.

Binomial coefficients and probabilities $\Bin_a(k)$ with large $a$ are enclosed as products of
rational factors, in a scaled representation $(\text{interval},\ \text{integer exponent of }2)$ that
avoids overflow and underflow; their logarithms are enclosed as the logarithm of the interval plus the
exponent.

\begin{proposition}[Soundness of the verified evaluation]
\label{prop:verified}
Let $\widetilde{\mathcal R}$ be either the expression~\eqref{eq:dm-trunc} or the right-hand side
of~\eqref{eq:vtr-lb} with $\Phi$ replaced by a finite partial sum, evaluated at the exact law
$P_z=f_z/\sum f$ and the exact $d$.
Let $[\underline{\mathcal R},\overline{\mathcal R}]$ be the interval obtained by evaluating the same
finite expression with the interval operations above, starting from intervals that contain $d$ and
each $f_z$.
Then $\underline{\mathcal R}\le\widetilde{\mathcal R}\le\Cdel$.
\end{proposition}

\begin{proof}
The expression is a finite composition of the operations $+$, $-$, $\times$, $\div$ (with divisors
bounded away from zero), $\ln$, and $\phi$.
We show by induction over the composition that every intermediate interval contains the exact value
of the corresponding subexpression.
The inputs are enclosed by construction.
If the arguments of an operation are enclosed, the output interval contains the exact result
because each interval operation contains the exact result for all arguments in the input intervals,
in particular for the exact arguments.
Hence the final interval contains $\widetilde{\mathcal R}$, and
$\underline{\mathcal R}\le\widetilde{\mathcal R}$.
For the DM functional, $\widetilde{\mathcal R}\le\Cdel$ is Proposition~\ref{prop:dm-trunc}.
For the VTR functional, replacing $\Phi$ by a partial sum of its nonnegative series decreases the
right-hand side of~\eqref{eq:vtr-lb}, which is at most $\Cdel$ by Theorem~\ref{thm:vtr}.
\end{proof}

\begin{algorithm}[!t]
\caption{Verified evaluation of a candidate law}
\label{alg:verify}
\small
\begin{algorithmic}[1]
\Require decimal $d$; candidate $f\in\R^Z_{>0}$ (doubles); truncations $R$ and window widths
  ($10$ std.\ dev.\ for $\Lambda$, $12$ for $\Theta$); context lengths $\kappa\in\{1,3\}$
\Ensure a number $\underline{\mathcal R}$ with $\underline{\mathcal R}\le\Cdel$
\State set the IEEE~754 rounding mode to $-\infty$; all operations below are interval operations
  (upper endpoints by negation)
\State $\boldsymbol d\gets[\text{prev.\ double},\text{next double}]\ni d$;\
  $\boldsymbol P_z\gets[f_z]/\bigl[\sum_{z'}f_{z'}\bigr]$ \Comment{exact law $f/\sum f$ is enclosed}
\State enclose $\mu$, $D$, and $\Psi_r$ by~\eqref{eq:psi-rec} for $r\le Z+R$
\If{DM functional}
  \State enclose $\Bin_a(k)$ by the Pascal recurrence and $\tilde q_k$ by the mixture~\eqref{eq:q-bilinear}, $a\le Z+R$
  \State lower-enclose $\tilde\Lambda$ over $r\le R$ and the $k$-window; enclose $\E[K]$ by~\eqref{eq:EK}
  \State $\underline{\mathcal R}\gets$ lower endpoint of~\eqref{eq:dm-trunc}
\Else\ (VTR functional)
  \State lower-enclose $\HH(P)/\mu$ and a partial sum $\tilde\Phi$ of~\eqref{eq:phi-def} with interval $\Theta_{z,r}$
  \State upper-enclose $\HH(\Omega)-\HH(\omega)$ and, for each $\kappa$, \eqref{eq:hsy-filter} with interval
    matrices $T^{\mathrm{same}}$, $T^{\mathrm{flip}}$ and $g$ from~\eqref{eq:g}
  \State $\underline{\mathcal R}\gets$ largest over $\kappa$ of the lower endpoints of~\eqref{eq:vtr-lb} with $\tilde\Phi$
\EndIf
\State \Return $\underline{\mathcal R}$ \Comment{tables show $\lfloor10^5\underline{\mathcal R}\rfloor/10^5$}
\end{algorithmic}
\end{algorithm}

Algorithm~\ref{alg:verify} summarizes the procedure.
The evaluated expressions are as follows.
For the DM functional, the enclosures cover $\mu$, $D$, $\Psi_r$ for $r\le Z+R$ (by the
recursion~\eqref{eq:psi-rec}), the partial sums $\tilde q_k$ of the mixture~\eqref{eq:q-bilinear}
over $a\le Z+R$ (with $\Bin_a(k)$ computed by the Pascal recurrence
$\Bin_a(k)=\db\,\Bin_{a-1}(k-1)+d\,\Bin_{a-1}(k)$), and the partial sum $\tilde\Lambda$ over $r\le R$ and
$k$ within ten standard deviations of the $\mathrm{Bin}(z+r,\db)$ law from its mean $\db(z+r)$; terms
are omitted only when they are nonnegative,
as Proposition~\ref{prop:dm-trunc} permits.
The prefactor uses the exact expression~\eqref{eq:EK}.
For the VTR functional, $\HH(P)/\mu$, $\HH(L\mid N)$, and~\eqref{eq:hsy-filter} for
$\kappa\in\{1,3\}$ are finite expressions and are enclosed without truncation; $g$ is enclosed from
its closed form~\eqref{eq:g}.
The partial sum of $\Phi$ runs over $r\le R$, and, in each $\Theta_{z,r}$, over $s$ within twelve standard
deviations of the $\mathrm{Bin}(z+r,\db)$ law from its mean, with the hypergeometric entropies enclosed term by term.
The reported value is the larger of the two lower endpoints obtained with $\kappa=1$ and $\kappa=3$,
each of which is a lower bound on $\Cdel$ by Proposition~\ref{prop:verified}.
Tables list the lower endpoints rounded down to five decimals.
The omitted mass $\varepsilon$ of Proposition~\ref{prop:dm-trunc} is below $10^{-12}$ in all DM
evaluations; this is reported for information and is not needed for validity.
The candidate laws, the retained index ranges, the enclosures of every term, and the verifier are
provided as supplementary material.

\subsection{Validation}
\label{subsec:validation}

The verifier and the formulas were checked in three ways.
First, for $d=0.2$ and a non-geometric law on $\{1,\dots,10\}$, the DM functional was evaluated
directly from the type probabilities~\eqref{eq:dm-joint}, without Lemma~\ref{lem:dm-reduce}, in
40-digit arithmetic, with omitted probability mass below $10^{-18}$.
The direct value is $0.326651950637153$, and the verified lower endpoint is $0.326651950637135$; the
closed form~\eqref{eq:EK} agrees with the direct mean to all printed digits.
Second, for $P=(\frac12,\frac12)$ and $d=\frac12$, where the residual chain has two states, the
enclosures of $\HH_\pi(S_t\mid Y_{t-1},Y_t)$ and $\HH_\pi(S_t\mid Y_{t-3},\dots,Y_t)$ contain
$0.95530631$ and $0.94835626$, the values obtained by exact enumeration of the chain of
Lemma~\ref{lem:residual}, independently of Proposition~\ref{prop:filter}.
Third, Lemma~\ref{lem:chain} was compared with Monte Carlo estimates of
$\frac1n\HH(S\mid X^n,Y)$ for a non-geometric law ($n=3000$, $20$ samples; the posterior of the
survivor counts given $X^n$ and $Y$ is computed exactly for each sample by a forward recursion over
input runs).
At $d=0.1$, $0.3$, and $0.6$ the bound $\Phi$ equals $\MCphiA$, $\MCphiB$, and $\MCphiC$, below the
estimates $\MCestA$, $\MCestB$, and $\MCestC$ (95\% confidence half-widths $\MCciA$, $\MCciB$,
$\MCciC$).
This comparison is a consistency check only; the difference reflects ambiguities that
Lemma~\ref{lem:chain} does not count (Section~\ref{sec:discussion}).

%% file: figures/pipeline_TIT.tex
\definecolor{archnavy}{RGB}{30,58,84}
\definecolor{archline}{RGB}{90,118,138}
\definecolor{archmuted}{RGB}{70,88,102}
\definecolor{archXbg}{RGB}{232,238,244}
\definecolor{archYbg}{RGB}{228,235,241}
\definecolor{archSbg}{RGB}{222,232,240}
\definecolor{archhead}{RGB}{36,68,96}
\definecolor{archbar}{RGB}{186,205,220}
\definecolor{archcert}{RGB}{214,234,222}
\begin{tikzpicture}[
  x=1cm,y=1cm,
  font=\scriptsize,
  arr/.style={-{Stealth[round,length=5.0pt,width=3.5pt]}, line width=0.65pt,
              color=archnavy!70!black, shorten >=0.9pt, shorten <=0.9pt},
  darr/.style={-{Stealth[round,length=5.0pt,width=3.5pt]}, line width=0.58pt, densely dashed,
               color=archnavy!60!black, shorten >=0.9pt, shorten <=0.9pt},
  blk/.style={draw=archline, line width=0.55pt, rounded corners=1.5pt, fill=white, align=center,
              inner xsep=3.2pt, inner ysep=2.4pt, font=\scriptsize, minimum height=0.92cm},
  dimlab/.style={font=\tiny, text=archmuted, inner sep=0.4pt},
]
\node[anchor=west, text=archhead, font=\scriptsize\bfseries] at (0,0.98)
  {(a)~Search variables};
\node[blk, text width=1.55cm] (din) at (1.05,0) {$d\in(0,1)$\\[-0.4pt]{\tiny deletion prob.}};
\node[blk, text width=2.65cm, right=0.55cm of din] (init)
  {$Z$, $R$, initial laws\\[-0.4pt]{\tiny truncated geometric, earlier laws}};
\node[blk, text width=2.05cm, right=0.55cm of init] (lam) {$\lambda\in\mathbb{R}^{Z}$\\[-0.4pt]{\tiny logits}};
\node[blk, text width=2.55cm, fill=archbar!35, right=0.55cm of lam] (P)
  {$P=\mathrm{softmax}(\lambda)$\\[-0.4pt]{\tiny $P\in\Delta_Z$ (relative interior)}};
\draw[arr] (din.east) -- (init.west);
\draw[arr] (init.east) -- (lam.west);
\draw[arr] (lam.east) -- (P.west);

\node[anchor=west, text=archhead, font=\scriptsize\bfseries] at (0,-1.08)
  {(b)~Two functionals};
\node[blk, text width=5.55cm, fill=archXbg] (dm) at (3.35,-2.62)
  {{\bfseries Drinea--Mitzenmacher functional}\\[1.2pt]
   tables $G_{z,r}$, $\Bin_a(k)$ built once per $d$ {\tiny (Lemma~\ref{lem:dm-reduce})}\\[1.4pt]
   $\mu,\ D,\ \Psi$ {\tiny\eqref{eq:psi-rec}},\ $\Lambda=P^{\!\top}G\Psi$,\ $\tilde q$ {\tiny\eqref{eq:q-bilinear}},\ $\E[K]$ {\tiny\eqref{eq:EK}}\\[0.8pt]
   $\widetilde{\mathcal R}(P,d)$ {\tiny\eqref{eq:dm-trunc}, Thm.~\ref{thm:dm}}};
\node[blk, text width=5.55cm, fill=archYbg, right=0.50cm of dm] (vtr)
  {{\bfseries Free-law VTR functional}\\[1.2pt]
   kernel $\Theta_{z,r}$ built once per $d$ {\tiny\eqref{eq:Kchain}}\\[1.4pt]
   $\HH(P)/\mu$,\ \ $\HH(L|N)$ {\tiny (Lemma~\ref{lem:hxsy})}\\[0.4pt]
   $\HH_\pi(S_t|Y_{t-\kappa}^{t})$ {\tiny (Prop.~\ref{prop:filter})},\ \ $\Phi$ {\tiny (Lemma~\ref{lem:chain})}\\[0.8pt]
   $\RVTR(P,d;\kappa)$ {\tiny\eqref{eq:vtr-lb}, Thm.~\ref{thm:vtr}}};
\coordinate (pdrop) at ($(P.south)+(0,-0.24)$);
\fill[archnavy] (pdrop) circle (1.15pt);
\draw[arr] (P.south) -- (pdrop);
\draw[arr] (pdrop) -| (dm.north);
\draw[arr] (pdrop) -| (vtr.north);

\node[anchor=west, text=archhead, font=\scriptsize\bfseries] at (0,-4.20)
  {(c)~Optimization loop};
\node[blk, text width=3.0cm, fill=archSbg] (grad) at (3.35,-5.50)
  {gradient $\nabla_\lambda$\\[0.4pt]
   {\tiny DM: reverse-mode AD; VTR: adjoint \eqref{eq:adjoint}--\eqref{eq:adjoint-grad}}};
\node[blk, text width=2.3cm, right=0.35cm of grad] (adam)
  {update $\lambda$\\[-0.4pt]{\tiny DM: Adam, rate $0.03$\\ VTR: L-BFGS}};
\node[blk, text width=2.6cm, right=0.35cm of adam] (cert)
  {monitor\\[-0.2pt]{\tiny DM: re-evaluate every $50$ steps\\ VTR: at termination}};
\node[blk, text width=2.4cm, fill=archbar!40, right=0.35cm of cert] (best)
  {keep best law\\[0.5pt] candidate $f\in\mathbb{R}^Z_{>0}$};
\coordinate (adentry) at ([xshift=0.9cm]grad.north);
\draw[arr] (dm.south) -- (grad.north);
\draw[arr] (vtr.south) -- ++(0,-0.35) -| (adentry);
\draw[arr] (grad.east) -- (adam.west);
\draw[arr] (adam.east) -- (cert.west);
\draw[arr] (cert.east) -- (best.west);
\coordinate (fbdown) at ($(adam.south)+(0,-0.30)$);
\coordinate (fbleft) at ([xshift=-0.62cm]din.west |- fbdown);
\coordinate (fbtop) at (fbleft |- 0,1.14);
\draw[darr] (adam.south) -- (fbdown) -- (fbleft) -- (fbtop) -| (lam.north);
\node[dimlab, anchor=south] at ([yshift=2.5pt]fbtop -| lam) {update $\lambda$};

\node[anchor=west, text=archhead, font=\scriptsize\bfseries] at (0,-7.05)
  {(d)~Verified evaluation};
\node[blk, text width=2.85cm, fill=archcert] (exact) at (1.70,-8.10)
  {exact rational law\\ $P_z=f_z/\sum_{z'}f_{z'}$;\ $d\in[\underline d,\overline d]$};
\node[blk, text width=3.30cm, fill=archcert, right=0.42cm of exact] (iv)
  {interval arithmetic\\[-0.2pt]{\tiny directed rounding, verified $\ln$,\\ scaled binomials}};
\node[blk, text width=2.95cm, fill=archcert, right=0.42cm of iv] (trunc)
  {one-sided truncation\\[-0.2pt]{\tiny Prop.~\ref{prop:dm-trunc}; partial sums of $\Phi$}};
\node[blk, text width=2.55cm, fill=archbar!55, right=0.42cm of trunc] (lb)
  {lower endpoint $\le C(d)$\\[-0.2pt]{\tiny (Prop.~\ref{prop:verified})\\ envelope $\max(\text{Thm.~\ref{thm:dm}},\text{Thm.~\ref{thm:vtr}})$}};
\coordinate (cdown) at ($(best.south)+(0,-0.62)$);
\coordinate (cleft) at ([xshift=-0.35cm]exact.west |- cdown);
\draw[arr] (best.south) -- (cdown) -- (cleft) |- (exact.west);
\draw[arr] (exact.east) -- (iv.west);
\draw[arr] (iv.east) -- (trunc.west);
\draw[arr] (trunc.east) -- (lb.west);
\end{tikzpicture}

%% file: sec_results_TIT.tex
\section{Numerical Results}
\label{sec:results}

All new values in this section are lower endpoints of interval enclosures
(Section~\ref{subsec:verified}), rounded down to five decimals.
Published values are quoted with the precision of the cited tables; values computed from published
formulas (linear bounds, $1-h(d)$) are rounded down.

\begin{figure*}[!t]
\centering
\includegraphics[width=\textwidth,height=0.76\textheight,keepaspectratio]{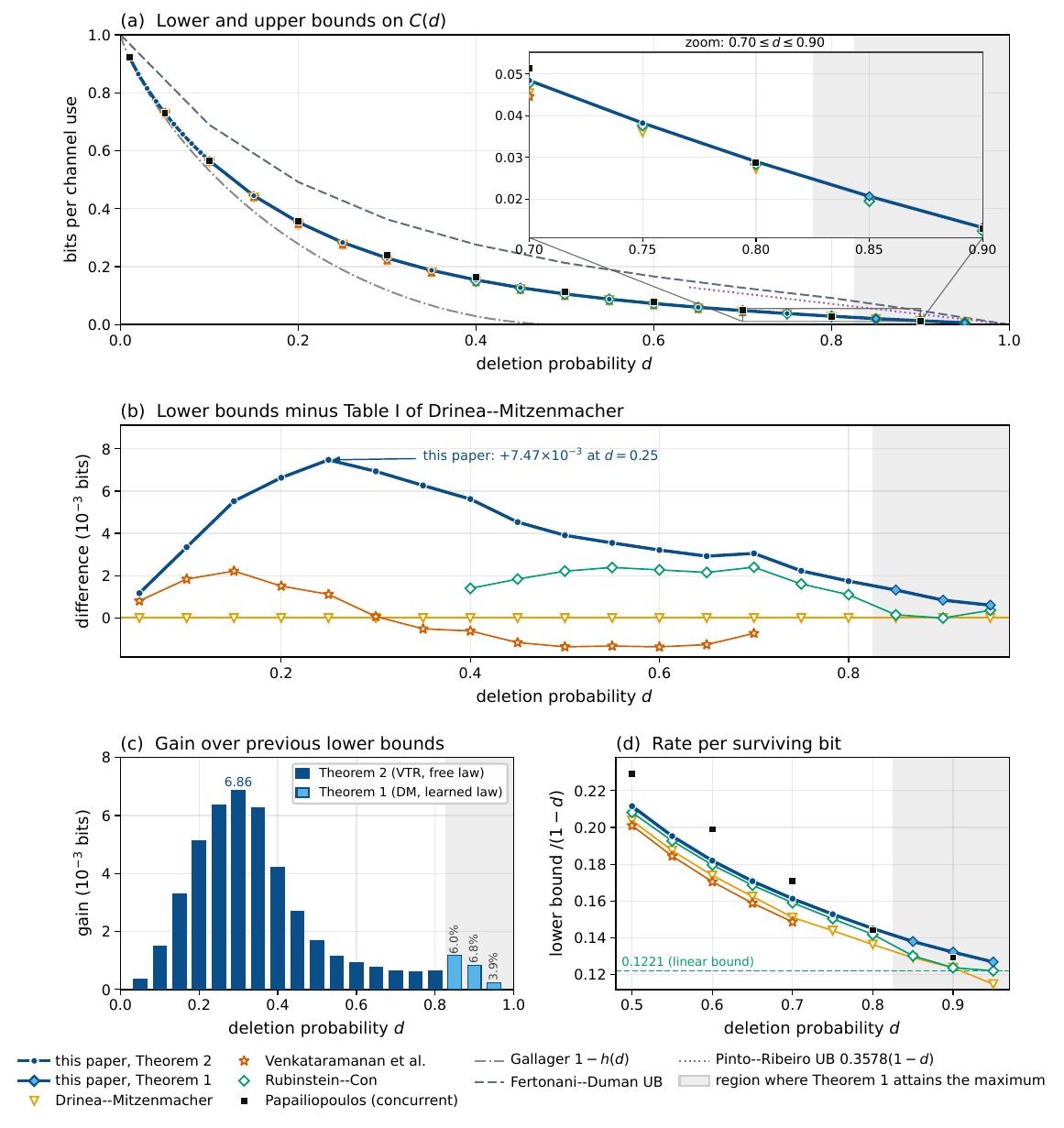}
\caption{Lower and upper bounds on $\Cdel$. Line segments between the values at tabulated $d$ are guides to the eye; the bounds are the values at those $d$ (and, for $1-h(d)$ and $0.3578(1-d)$, the curves themselves).
The present bound is the larger of Theorems~\ref{thm:dm} and~\ref{thm:vtr} (verified lower
endpoints): dark circles mark values of $d$ where Theorem~\ref{thm:vtr} attains the maximum, light
diamonds those where Theorem~\ref{thm:dm} does, and the shaded region is $d\ge\DmFrom$, where
Theorem~\ref{thm:dm} attains it.
(a)~Lower and upper bounds; the inset magnifies $d\in[0.70,0.90]$.
Open markers are the tabulated bounds
of~\cite{DrineaMitzenmacher2007,VenkataramananTatikondaRamchandran2013,RubinsteinCon2023}
(Rubinstein--Con: larger of their Table~2 and $0.1221(1-d)$); filled squares are the lower endpoints
of the concurrent enclosure~\cite{Papailiopoulos2026}.
(b)~The same lower bounds minus Table~I of~\cite{DrineaMitzenmacher2007}, in $10^{-3}$~bits.
(c)~Gain over the largest previously published lower bound at the $\NumTab$ values of $d$ with
tabulated run-length bounds (the baseline of Table~\ref{tab:comparison}); the largest gain is labelled,
and relative gains are given for $d\ge\DmFrom$, where the absolute gains are small.
(d)~Lower bounds divided by $1-d$ (rate per surviving bit); the dashed line is the linear bound
$0.1221(1-d)$ of~\cite{RubinsteinCon2023}.
Upper bounds:~\cite{FertonaniDuman2010,PintoRibeiro2026ParallelBA}.}
\label{fig:comparison}
\end{figure*}

\subsection{The Drinea--Mitzenmacher functional}

Table~\ref{tab:dm} compares three evaluations of Theorem~\ref{thm:dm}: the published values of
Drinea and Mitzenmacher for geometric laws with two-decimal stay probabilities, a geometric control in
which the stay probability $p$ is optimized continuously on the same support as the learned law, and
the learned law.
The control separates the gain due to a free law from the gain due to a finer choice of $p$.
For $d\le0.03$ the learned law is close to geometric and the gain over the control is below
$10^{-4}$~bits.
The gain is largest, $\DmFreeMax\times10^{-3}$~bits, at $d=\DmFreeMaxD$; in relative terms the learned
law exceeds the optimized geometric law by $\DmFreeRelNine\%$ at $d=0.9$.

\begin{table}[!t]
\caption{Theorem~\ref{thm:dm} for geometric and learned run-length laws.
``Geometric'' is the truncated geometric law on the support of the learned law with the stay
probability $p^\star$ optimized continuously; ``Gain'' is learned minus geometric.
All bound values except the published column are verified lower endpoints, rounded down.}
\label{tab:dm}
\centering
\footnotesize
\setlength{\tabcolsep}{4pt}
\input{data_TIT/tab_dm.tex}
\end{table}

\begin{figure*}[!t]
\centering
\includegraphics[width=\textwidth]{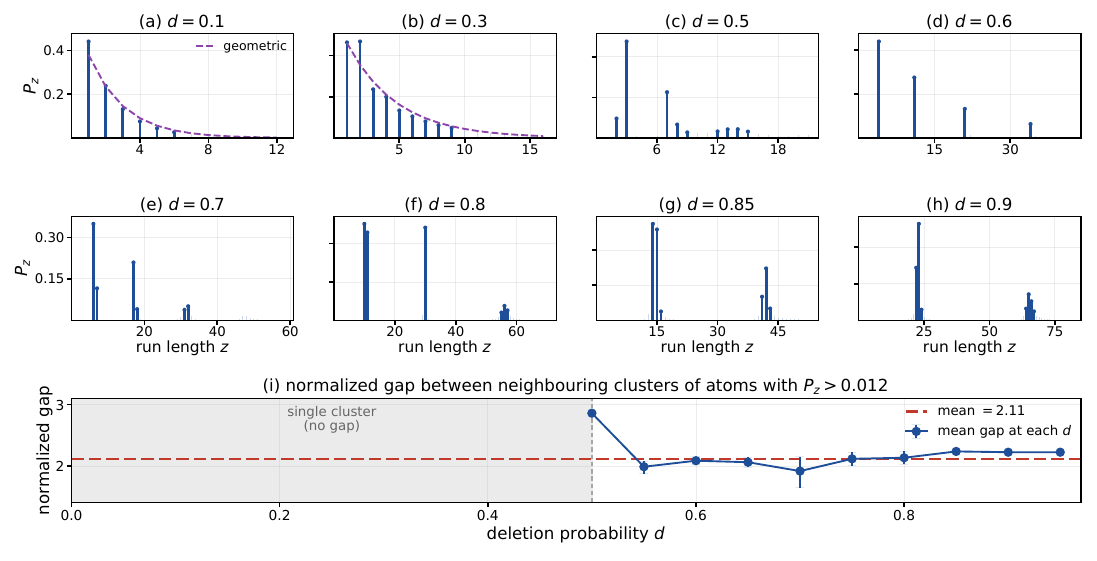}
\caption{Learned run-length laws for Theorem~\ref{thm:dm}.
(a)--(h) $P_z$ at increasing $d$; markers indicate atoms with $P_z\ge0.02$, and the dashed curve is
the geometric law of~\cite[Table~I]{DrineaMitzenmacher2007}.
(i) Normalized gaps~\eqref{eq:gap} between neighbouring clusters of atoms with $P_z>0.012$ for
$d\ge0.5$ (marker: mean at that $d$; whiskers: minimum and maximum); the horizontal line is the mean
over all gaps.}
\label{fig:shape}
\end{figure*}

Figure~\ref{fig:shape} shows how the learned laws change with $d$.
Up to $d\approx0.3$ the law is unimodal and close to geometric; for $0.35\le d\le0.45$ it develops
secondary modes.
From $d=0.5$ on, most of the mass lies on a few separated clusters of neighbouring run lengths: the
atoms with $P_z>0.012$ carry between $\ClusterMassMin\%$ and $\ClusterMassMax\%$ of the mass for
$0.5\le d\le0.8$ (and more than $55\%$ up to $d=0.95$).
At $d=0.8$ the clusters have centroids $\CentEight$, and the atoms $z\le3$ carry less than
$10^{-4}$ of the mass.
To describe the spacing, group the atoms with $P_z>0.012$ into clusters of indices at most two apart,
let $c_1<c_2<\cdots$ be the cluster centroids (weighted by $P_z$), and let
$\sigma(c)=\sqrt{c\,\db d}$ be the standard deviation of the number of survivors of a run of length
$c$.
The normalized gap between neighbouring clusters is
\begin{equation}
\label{eq:gap}
\frac{\db(c_{j+1}-c_j)}{\frac12\bigl(\sigma(c_j)+\sigma(c_{j+1})\bigr)}
=2\sqrt{\frac{\db}{d}}\bigl(\sqrt{c_{j+1}}-\sqrt{c_j}\bigr).
\end{equation}
Over the $\NGaps$ gaps observed for $d\ge0.5$ its mean is $\MeanGap$ and its range is
$[\GapMin,\GapMax]$; at $d=0.8$ the two gaps are $\GapEight$.
By~\eqref{eq:gap}, a constant normalized gap is the same as equal spacing of the square roots of the
centroids.
A heuristic reading is that the number of survivors of an isolated run of length $c$ has mean $\db c$
and standard deviation $\sigma(c)$, and that runs whose survivor counts are about two standard
deviations apart are distinguishable by the decoder, analogous to the points of a pulse-amplitude
constellation.
The reading is only heuristic: an output run can merge several input runs, and its length has the
mixture law~\eqref{eq:q-bilinear} rather than a single binomial law.
The cluster structure is an empirical property of the laws found by the search; neither the geometric
family nor the two-parameter family of~\cite{DrineaMitzenmacher2007} contains such laws.

In a support study for the DM search~\cite{khodaiemehr2026improvedlowerboundscapacity} at $d\in\{0.1,0.3,0.5,0.8\}$, the learned law was padded with
zeros to supports $Z/\mu_0\in\{5,10,15,20,25,30,40\}$ and $Z=\max(1000,40\mu_0)$, and re-optimized
at most of these supports; Fig.~\ref{fig:support} shows the double-precision values.
Enlarging the support beyond $Z\approx20\mu_0$ changed the value by at most $1.4\times10^{-10}$~bits
at these values of $d$.
This is an observation about the search protocol; it does not bound what other laws with larger
support can achieve.

\begin{figure*}[!t]
\centering
\includegraphics[width=0.85\textwidth]{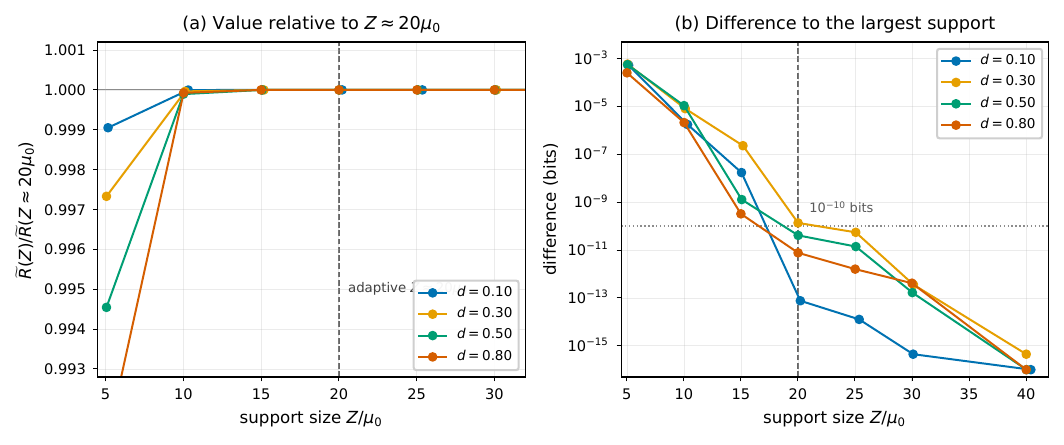}
\caption{Support study for the search of Algorithm~\ref{alg:dm} (double-precision values of the
retained expression~\eqref{eq:dm-trunc}).
(a) Value at support $Z$ divided by the value at $Z\approx20\mu_0$ (dashed line: the support used in
the main runs).
(b) Difference between the value at the largest tested support and the value at $Z$; differences
below $10^{-15}$ are at the level of double-precision rounding.}
\label{fig:support}
\end{figure*}

\subsection{The VTR functional}

Table~\ref{tab:vtr} lists Theorem~\ref{thm:vtr} for a Markov control (truncated geometric laws with the
stay probability $\gamma^\star$ optimized continuously) and for the law found by the search.
With truncated geometric laws, Theorem~\ref{thm:vtr} exceeds the published VTR values by
$\MarkovGainMin\times10^{-3}$ to $\MarkovGainMax\times10^{-3}$~bits; most of this gain is due to the
chain penalty (Proposition~\ref{prop:vtr-compare}).
A free run-length law adds up to $\FreeGainMax\times10^{-3}$~bits more.
The free-law VTR bound exceeds the DM bound with a learned law for $d\le\VtrUpTo$, and the DM bound is
larger from $d=\DmFrom$ on.
The functional was not evaluated at $d=0.95$.
For the laws found by the search, the context length $\kappa=3$ in~\eqref{eq:vtr-lb} improves on
$\kappa=1$ by up to $\KappaGain\times10^{-4}$~bits.

\begin{table}[!t]
\caption{Theorem~\ref{thm:vtr} for truncated geometric laws (Markov control) and for the law found by the search, compared with
the published VTR values and with Theorem~\ref{thm:dm} at the learned law.
The larger of the last two columns is in bold.}
\label{tab:vtr}
\centering
\footnotesize
\setlength{\tabcolsep}{3.5pt}
\input{data_TIT/tab_vtr.tex}
\end{table}

\subsection{Comparison with published bounds}

Table~\ref{tab:comparison} and Figs.~\ref{fig:comparison}--\ref{fig:gain} compare the larger of the two bounds with the
largest previously published lower bound at every $d$ for which a run-length bound has been
tabulated.
The baseline is the maximum of Tables~I and~II of~\cite{DrineaMitzenmacher2007}, Table~I
of~\cite{VenkataramananTatikondaRamchandran2013}, Table~2 of~\cite{RubinsteinCon2023}, the linear bounds
$0.1185(1-d)$~\cite{MitzenmacherDrinea2006} and $0.1221(1-d)$~\cite{RubinsteinCon2023}, and $1-h(d)$ for
$d<1/2$.
The new bounds are larger at all $\NumTab$ such values of $d$ (Fig.~\ref{fig:comparison}(c)).
Fig.~\ref{fig:comparison}(b) shows that, relative to Table~I of~\cite{DrineaMitzenmacher2007}, the gain
is largest at moderate $d$, where Theorem~\ref{thm:vtr} applies, and Fig.~\ref{fig:gain}(b) shows the
relative gains: up to $\RelGainVal\%$ over the previous best at $d=\RelGainD$, and, for each
functional, the gain over the published evaluation at a geometric (Markov) law grows with $d$.
The largest absolute gain is $\AbsGainVal$~bits at $d=\AbsGainD$, and the largest relative gain is
$\RelGainVal\%$ at $d=\RelGainD$.
At the values of $d$ where no run-length bound was tabulated ($d<0.05$ and $0.05<d<0.1$), the only
published bounds are $1-h(d)$ and the linear bounds $0.1185(1-d)$ and $0.1221(1-d)$, which the new
values exceed; the geometric controls of Tables~\ref{tab:dm} and~\ref{tab:vtr} are the appropriate
reference there.
The Kirsch--Drinea~\cite{KirschDrinea2009} and Castiglione--Kav\v{c}i\'c~\cite{CastiglioneKavcic2015}
values rely on Monte Carlo estimates and are not included.

The concurrent enclosure of Papailiopoulos~\cite{Papailiopoulos2026} has larger lower endpoints at
$d\in\{\PPbetter\}$, and the present bounds are larger at $d\in\{\WeBetter\}$.
Every new value lies below the corresponding upper endpoint of~\cite{Papailiopoulos2026}
(Fig.~\ref{fig:bounds}).
The two approaches are complementary: the enclosure of~\cite{Papailiopoulos2026} evaluates specified
finite-state and independent-run sources and also supplies converse bounds, while the present work
optimizes the run-length law inside two classical functionals.

\begin{table}[!t]
\caption{Published lower bounds on $\Cdel$ and this paper.
D--M:~\cite{DrineaMitzenmacher2007} (Table~I);
VTR:~\cite{VenkataramananTatikondaRamchandran2013} (Table~I);
R--C:~\cite{RubinsteinCon2023} (larger of Table~2 and $0.1221(1-d)$);
P--P: concurrent enclosure of~\cite{Papailiopoulos2026} (Table~3 lower endpoints, rounded down);
$1{-}h(d)$:~\cite{Gallager1961} ($d<1/2$ only).
The last column is the larger of Theorems~\ref{thm:dm} and~\ref{thm:vtr} (verified lower endpoints,
rounded down); ${}^{\ast}$: attained by Theorem~\ref{thm:vtr}, otherwise by Theorem~\ref{thm:dm}
(Theorem~\ref{thm:vtr} was not evaluated at $d=0.95$).
${}^{\dagger}$: $d$ not listed in Table~I of~\cite{DrineaMitzenmacher2007}; the entry is our verified
evaluation of their functional at the geometric law with optimized stay probability
(Table~\ref{tab:dm}), not a published value.
Values from published formulas and from other papers are rounded down to five decimals (four for VTR,
as printed there).
The largest entry in each row is in bold.}
\label{tab:comparison}
\centering
\scriptsize
\setlength{\tabcolsep}{2.6pt}
\resizebox{\tabIVwidth}{!}{\input{data_TIT/tab_comparison.tex}}
\end{table}

\begin{figure*}[!t]
\centering
\includegraphics[width=\textwidth]{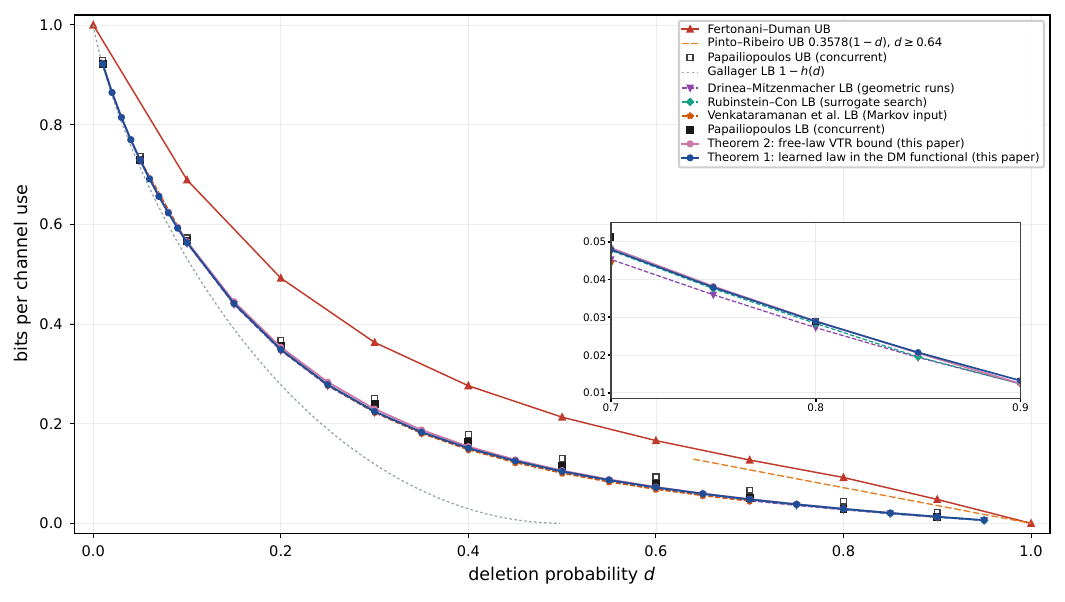}
\caption{Lower and upper bounds on $\Cdel$, with the two theorems shown separately. Line segments between the values at tabulated $d$ are guides to the eye; the bounds are the values at those $d$ (and, for $1-h(d)$ and $0.3578(1-d)$, the curves themselves).
Blue: Theorem~\ref{thm:dm} at the learned run-length law; pink: Theorem~\ref{thm:vtr} at the law found
by the search ($d\le0.9$).
The inset magnifies $d\in[0.70,0.90]$, where the curves otherwise overlap.
Tabulated lower bounds:~\cite{DrineaMitzenmacher2007,VenkataramananTatikondaRamchandran2013,RubinsteinCon2023};
lower and upper endpoints of the concurrent enclosure:~\cite{Papailiopoulos2026}.
Upper bounds:~\cite{FertonaniDuman2010,PintoRibeiro2026ParallelBA,Papailiopoulos2026}.}
\label{fig:bounds}
\end{figure*}

\begin{figure*}[!t]
\centering
\includegraphics[width=\textwidth]{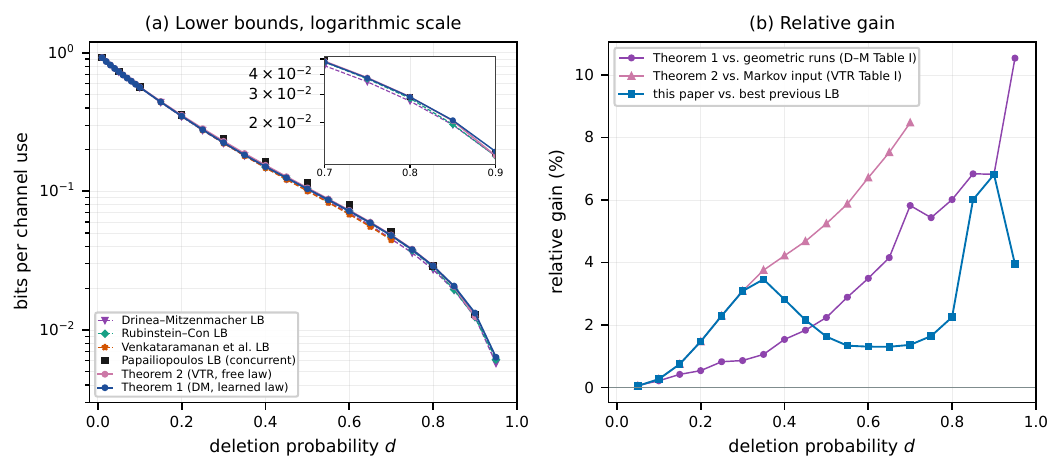}
\caption{(a)~Lower bounds on a logarithmic scale; the inset magnifies $d\in[0.70,0.90]$. Line segments between the values at tabulated $d$ are guides to the eye.
(b)~Relative gains: Theorem~\ref{thm:dm} at the learned law over the geometric law of Table~I
of~\cite{DrineaMitzenmacher2007}; Theorem~\ref{thm:vtr} at the law found by the search over the
Markov-input values of Table~I of~\cite{VenkataramananTatikondaRamchandran2013}; and the larger of the
two theorems over the largest previously published lower bound (the baseline of
Table~\ref{tab:comparison}).}
\label{fig:gain}
\end{figure*}

%% file: data_TIT/tab_dm.tex
\begin{tabular}{@{}cccccc@{}}
\toprule
$d$ & D--M Tab.~I & $p^\star$ & Geometric & Learned & Gain \\
 & \cite{DrineaMitzenmacher2007} & & (opt.\ $p$) & law & ($10^{-3}$) \\
\midrule
0.01 & --- & 0.5062 & 0.92210 & 0.92211 & 0.01 \\
0.02 & --- & 0.5127 & 0.86440 & 0.86444 & 0.04 \\
0.03 & --- & 0.5195 & 0.81444 & 0.81452 & 0.08 \\
0.04 & --- & 0.5267 & 0.76957 & 0.76972 & 0.15 \\
0.05 & 0.72829 & 0.5342 & 0.72856 & 0.72879 & 0.22 \\
0.06 & --- & 0.5419 & 0.69067 & 0.69098 & 0.31 \\
0.07 & --- & 0.5500 & 0.65541 & 0.65582 & 0.42 \\
0.08 & --- & 0.5584 & 0.62243 & 0.62296 & 0.53 \\
0.09 & --- & 0.5671 & 0.59150 & 0.59214 & 0.64 \\
0.10 & 0.56196 & 0.5761 & 0.56239 & 0.56316 & 0.77 \\
0.15 & 0.43918 & 0.6241 & 0.43964 & 0.44102 & 1.38 \\
0.20 & 0.34669 & 0.6749 & 0.34674 & 0.34855 & 1.82 \\
0.25 & 0.27588 & 0.7243 & 0.27618 & 0.27815 & 1.97 \\
0.30 & 0.22243 & 0.7689 & 0.22243 & 0.22434 & 1.91 \\
0.35 & 0.18101 & 0.8071 & 0.18108 & 0.18292 & 1.84 \\
0.40 & 0.14841 & 0.8391 & 0.14874 & 0.15069 & 1.95 \\
0.45 & 0.12286 & 0.8655 & 0.12292 & 0.12511 & 2.19 \\
0.50 & 0.10186 & 0.8875 & 0.10188 & 0.10414 & 2.26 \\
0.55 & 0.084323 & 0.9061 & 0.08438 & 0.08675 & 2.37 \\
0.60 & 0.069564 & 0.9220 & 0.06958 & 0.07199 & 2.41 \\
0.65 & 0.056858 & 0.9358 & 0.05685 & 0.05922 & 2.36 \\
0.70 & 0.045324 & 0.9480 & 0.04577 & 0.04796 & 2.19 \\
0.75 & 0.035984 & 0.9588 & 0.03599 & 0.03793 & 1.94 \\
0.80 & 0.027266 & 0.9685 & 0.02729 & 0.02890 & 1.61 \\
0.85 & 0.01938 & 0.9774 & 0.01947 & 0.02070 & 1.23 \\
0.90 & 0.012378 & 0.9855 & 0.01238 & 0.01322 & 0.83 \\
0.95 & 0.005741 & 0.9930 & 0.00592 & 0.00634 & 0.42 \\
\bottomrule
\end{tabular}

%% file: data_TIT/tab_vtr.tex
\begin{tabular}{@{}cccccc@{}}
\toprule
$d$ & VTR Tab.~I & \multicolumn{2}{c}{Trunc.\ geometric (Thm.~\ref{thm:vtr})} & Free law & DM \\
\cmidrule(lr){3-4}
 & \cite{VenkataramananTatikondaRamchandran2013} & $\gamma^\star$ & value & (Thm.~\ref{thm:vtr}) & (learned) \\
\midrule
0.01 & --- & 0.5062 & 0.92214 & \textbf{0.92215} & 0.92211 \\
0.02 & --- & 0.5126 & 0.86452 & \textbf{0.86456} & 0.86444 \\
0.03 & --- & 0.5193 & 0.81470 & \textbf{0.81478} & 0.81452 \\
0.04 & --- & 0.5262 & 0.77002 & \textbf{0.77017} & 0.76972 \\
0.05 & 0.7291 & 0.5334 & 0.72923 & \textbf{0.72946} & 0.72879 \\
0.06 & --- & 0.5408 & 0.69158 & \textbf{0.69191} & 0.69098 \\
0.07 & --- & 0.5485 & 0.65660 & \textbf{0.65703} & 0.65582 \\
0.08 & --- & 0.5564 & 0.62391 & \textbf{0.62447} & 0.62296 \\
0.09 & --- & 0.5646 & 0.59328 & \textbf{0.59396} & 0.59214 \\
0.10 & 0.5638 & 0.5730 & 0.56448 & \textbf{0.56530} & 0.56316 \\
0.15 & 0.4414 & 0.6178 & 0.44314 & \textbf{0.44469} & 0.44102 \\
0.20 & 0.3482 & 0.6654 & 0.35115 & \textbf{0.35332} & 0.34855 \\
0.25 & 0.277 & 0.7127 & 0.28088 & \textbf{0.28335} & 0.27815 \\
0.30 & 0.2225 & 0.7565 & 0.22695 & \textbf{0.22936} & 0.22434 \\
0.35 & 0.1805 & 0.7951 & 0.18517 & \textbf{0.18727} & 0.18292 \\
0.40 & 0.1478 & 0.8280 & 0.15233 & \textbf{0.15403} & 0.15069 \\
0.45 & 0.1217 & 0.8557 & 0.12602 & \textbf{0.12739} & 0.12511 \\
0.50 & 0.1005 & 0.8790 & 0.10452 & \textbf{0.10576} & 0.10414 \\
0.55 & 0.083 & 0.8987 & 0.08663 & \textbf{0.08787} & 0.08675 \\
0.60 & 0.0682 & 0.9157 & 0.07147 & \textbf{0.07277} & 0.07199 \\
0.65 & 0.0556 & 0.9305 & 0.05842 & \textbf{0.05978} & 0.05922 \\
0.70 & 0.0446 & 0.9435 & 0.04704 & \textbf{0.04837} & 0.04796 \\
0.75 & --- & 0.9549 & 0.03696 & \textbf{0.03821} & 0.03793 \\
0.80 & --- & 0.9652 & 0.02789 & \textbf{0.02901} & 0.02890 \\
0.85 & --- & 0.9748 & 0.01951 & 0.02052 & \textbf{0.02070} \\
0.90 & --- & 0.9847 & 0.01147 & 0.01245 & \textbf{0.01322} \\
\bottomrule
\end{tabular}

%% file: data_TIT/tab_comparison.tex
\begin{tabular}{@{}c@{\hspace{0.38em}}*{6}{c@{\hspace{0.36em}}}@{}}
\toprule
$d$ & $1{-}h(d)$ & D--M & VTR & R--C & P--P & This work \\
\midrule
0.01 & 0.91920 & 0.92210$^\dagger$ & --- & --- & 0.92211 & \textbf{0.92215}$^{\ast}$ \\
0.02 & 0.85855 & 0.86440$^\dagger$ & --- & --- & --- & \textbf{0.86456}$^{\ast}$ \\
0.03 & 0.80560 & 0.81444$^\dagger$ & --- & --- & --- & \textbf{0.81478}$^{\ast}$ \\
0.04 & 0.75770 & 0.76957$^\dagger$ & --- & --- & --- & \textbf{0.77017}$^{\ast}$ \\
0.05 & 0.71360 & 0.72829 & 0.7291 & --- & \textbf{0.72983} & 0.72946$^{\ast}$ \\
0.06 & 0.67255 & 0.69067$^\dagger$ & --- & --- & --- & \textbf{0.69191}$^{\ast}$ \\
0.07 & 0.63407 & 0.65541$^\dagger$ & --- & --- & --- & \textbf{0.65703}$^{\ast}$ \\
0.08 & 0.59782 & 0.62243$^\dagger$ & --- & --- & --- & \textbf{0.62447}$^{\ast}$ \\
0.09 & 0.56353 & 0.59150$^\dagger$ & --- & --- & --- & \textbf{0.59396}$^{\ast}$ \\
0.10 & 0.53100 & 0.56196 & 0.5638 & --- & \textbf{0.56660} & 0.56530$^{\ast}$ \\
0.15 & 0.39015 & 0.43918 & 0.4414 & --- & --- & \textbf{0.44469}$^{\ast}$ \\
0.20 & 0.27807 & 0.34669 & 0.3482 & --- & \textbf{0.35674} & 0.35332$^{\ast}$ \\
0.25 & 0.18872 & 0.27588 & 0.2770 & --- & --- & \textbf{0.28335}$^{\ast}$ \\
0.30 & 0.11870 & 0.22243 & 0.2225 & --- & \textbf{0.23903} & 0.22936$^{\ast}$ \\
0.35 & 0.06593 & 0.18101 & 0.1805 & --- & --- & \textbf{0.18727}$^{\ast}$ \\
0.40 & 0.02904 & 0.14841 & 0.1478 & 0.14981 & \textbf{0.16350} & 0.15403$^{\ast}$ \\
0.45 & 0.00722 & 0.12286 & 0.1217 & 0.12470 & --- & \textbf{0.12739}$^{\ast}$ \\
0.50 & --- & 0.10186 & 0.1005 & 0.10407 & \textbf{0.11454} & 0.10576$^{\ast}$ \\
0.55 & --- & 0.08432 & 0.0830 & 0.08671 & --- & \textbf{0.08787}$^{\ast}$ \\
0.60 & --- & 0.06956 & 0.0682 & 0.07183 & \textbf{0.07960} & 0.07277$^{\ast}$ \\
0.65 & --- & 0.05685 & 0.0556 & 0.05901 & --- & \textbf{0.05978}$^{\ast}$ \\
0.70 & --- & 0.04532 & 0.0446 & 0.04772 & \textbf{0.05127} & 0.04837$^{\ast}$ \\
0.75 & --- & 0.03598 & --- & 0.03759 & --- & \textbf{0.03821}$^{\ast}$ \\
0.80 & --- & 0.02726 & --- & 0.02837 & 0.02884 & \textbf{0.02901}$^{\ast}$ \\
0.85 & --- & 0.01938 & --- & 0.01953 & --- & \textbf{0.02070} \\
0.90 & --- & 0.01237 & --- & 0.01237 & 0.01292 & \textbf{0.01322} \\
0.95 & --- & 0.00574 & --- & 0.00610 & --- & \textbf{0.00634} \\
\bottomrule
\end{tabular}

%% file: sec_discussion_TIT.tex
\section{Discussion and Open Problems}
\label{sec:discussion}

The validity of every reported number rests on Theorem~\ref{thm:dm}~\cite{DrineaMitzenmacher2007},
on Lemmas~\ref{lem:psi}--\ref{lem:dm-reduce} and Proposition~\ref{prop:dm-trunc} for the DM
functional, on Theorem~\ref{thm:vtr} for the VTR functional, and on
Proposition~\ref{prop:verified} for the evaluation.
The search of Section~\ref{subsec:search} only selects the law at which the bounds are evaluated.
Statements about the shape of the laws found by the search, about support sizes, and about the
relative performance of search variants are observations about the search; they are not needed for
validity and do not imply that the laws are optimal.

\subsection{Tighter run-length functionals}

The Monte Carlo comparison of Section~\ref{subsec:validation} shows that $\Phi$ is about $70\%$ of
the estimated value of $\frac1n\HH(S\mid X^n,Y)$ at $d\in\{0.1,0.3,0.6\}$.
The chain penalty of Lemma~\ref{lem:chain} charges only the split of the survivors of an output run
among the runs of its chain, given the anchor, the number $i$ of deleted runs of the other symbol
inside the chain, and the survivor counts of all runs outside the chain (the variable $\Gamma_j$ in Step~2
of the proof).
The uncertainty about these revealed quantities, that is, about where the input runs that contribute
to one output run end and those that contribute to the next begin, is not counted.
A bound that reveals less, for example only the survivor counts of the runs of the same symbol
outside the chain, would lead to a correction term at least as large as $\Phi$, and it might still be evaluated with the renewal
weights of Lemma~\ref{lem:psi}.
Similarly, the context length $\kappa$ in Lemma~\ref{lem:hsy} can be increased at a cost of $2^\kappa$
matrix--vector products; for the laws found by the search, $\kappa=3$ improved on $\kappa=1$ by at most
$\KappaGain\times10^{-4}$~bits, which suggests, but does not prove, that the second term of~\eqref{eq:vtr-split} is already
close to its limit $\lim_\kappa\HH_\pi(S_t\mid Y_{t-\kappa},\dots,Y_t)$.

\subsection{The search}

Three questions about the search remain open.
First, the support study of Fig.~\ref{fig:support} shows that, for the DM functional at four values of
$d$, enlarging the support beyond $Z\approx20\mu_0$ changes the retained expression by at most
$1.4\times10^{-10}$~bits.
Whether maximizers exist over all laws with finite mean, and whether their tails decay geometrically
with a rate determined by $d$, is not known.
Second, in the earlier version of this work~\cite{khodaiemehr2026improvedlowerboundscapacity} we tested structured parameterizations of the DM search: finite mixtures
$P=\sum_{j}\alpha_jP^{(j)}$ of geometric laws and discretized Gaussian components centred at preferred
run lengths, trained with annealed sparsity penalties on $\HH(P)$ and $\HH(\alpha)$, and sparse
projections that keep the largest masses followed by free re-optimization.
On the grid of Table~\ref{tab:dm} these variants matched the laws of Algorithm~\ref{alg:dm} to within
$4\times10^{-6}$~bits in double precision.
Third, exponentiated-gradient (mirror-descent) updates
$P\leftarrow P\odot\exp(\eta\nabla_P\RDM)/\langle P,\exp(\eta\nabla_P\RDM)\rangle$ from comb-shaped
initial laws reached values $10^{-4}$ to $10^{-3}$~bits below those of Algorithm~\ref{alg:dm}.
These observations indicate that, for the DM functional, the remaining gap to the VTR functional at small
$d$ is a gap between functionals rather than a limitation of the search, which is consistent with
Table~\ref{tab:vtr}.

\subsection{Run-length constellations}

The laws found for the DM functional change from approximately geometric laws at small $d$ to laws
concentrated on a few clusters of run lengths at large $d$ (Fig.~\ref{fig:shape}).
The normalized gaps~\eqref{eq:gap} suggest that the clusters are arranged so that the numbers of
survivors of runs from neighbouring clusters are separated by about two binomial standard deviations,
as in a pulse-amplitude constellation whose noise grows with the amplitude.
This is an empirical observation, and the following questions are open.
\begin{enumerate}
\item Do maximizers of $\RDM(\cdot,d)$ or $\RVTR(\cdot,d;\kappa)$ put vanishing mass on short runs as
$d\to1$?
\item Must maximizers become multimodal for all sufficiently large $d$?
\item Is the normalized spacing~\eqref{eq:gap} of neighbouring clusters asymptotically constant, and if
so, what is the constant?
\item Can maximizers be characterized by first-order optimality conditions on a finite set of active
run lengths, as in the discrete input laws that achieve the capacity of amplitude-constrained
channels?
\end{enumerate}
An answer to any of these questions would turn the cluster structure into a property of the
functionals rather than of the search.

\subsection{Other channels}

Run-length inputs are natural for other synchronization channels, and the reductions of this paper
may serve as a starting point for channels that preserve the run structure used here: the channel acts
independently on the bits, never splits an input run, and produces survivor counts whose conditional
split given their total has a known law.
Channels with insertions can split a run (a single $0$ can become $01$), and they require a new
derivation.
For channels with i.i.d.\ duplications or insertions, the jigsaw and residual-run arguments
of~\cite{DrineaMitzenmacher2007,VenkataramananTatikondaRamchandran2013} have analogues, and recent work
gives small-insertion expansions and finite-blocklength
bounds~\cite{TeginDuman2025InsertionISIT,TeginDuman2026InsertionCap,MorozovDuman2026FiniteLength}.
Whether the renewal reduction, the chain penalty, and the verified search lead to improved bounds for
those channels is left for future work.